\documentclass[romanappendices, letterpaper, onecolumn]{IEEEtran}

\usepackage{xspace}
\usepackage{mathrsfs}
\usepackage{amsopn}

\newcommand{\Bc}{\mathcal{B}}

\newcommand{\Nc}{\mathcal{N}}

\newcommand{\Rc}{\mathcal{R}}
\newcommand{\Sc}{\mathcal{S}}

\newcommand{\Uc}{\mathcal{U}}
\newcommand{\Vc}{\mathcal{V}}

\newcommand{\Xc}{\mathcal{X}}
\newcommand{\Yc}{\mathcal{Y}}

\newcommand{\aep}{{\mathcal{T}_{\epsilon}^{(n)}}}

\newcommand{\Sh}{{\hat{S}}}

\newcommand{\Zh}{{\hat{Z}}}

\newcommand{\sh}{{\hat{s}}}
\newcommand{\uh}{{\hat{u}}}

\newcommand{\xh}{{\hat{x}}}

\newcommand{\St}{{\tilde{S}}}
\newcommand{\Ut}{{\tilde{U}}}

\def\d{\delta}
\def\e{\epsilon}

\DeclareMathOperator\E{E}
\let\P\relax
\DeclareMathOperator\P{P}

\newcommand{\Bern}{\mathrm{Bern}}

\def\textiid{i.i.d.\@\xspace}
\newcommand\iid{\ifmmode\text{ i.i.d. } \else \textiid \fi}

\usepackage{cite}
\usepackage{version}
\usepackage{fullpage,epsfig,graphicx,psfrag,color,subfig}
\usepackage{amsmath,amsthm}
\usepackage{amssymb}
\usepackage{microtype}
\DisableLigatures{encoding = *, family = * }

\newtheorem{theorem}{Theorem}
\newtheorem{proposition}{Proposition}
\newtheorem{corollary}{Corollary}
\newtheorem{remark}{Remark}[section]
\newtheorem{claim}{Claim}

\begin{document}
\title{Estimation with a helper who knows the interference}

\author{Yeow-Khiang Chia\IEEEauthorrefmark{1}, Rajiv Soundararajan\IEEEauthorrefmark{2} and Tsachy Weissman\IEEEauthorrefmark{3}
\thanks{This work was partially supported by NSF under grant CCF-0916713, AFOSR under grant FA95500910063, and the Center for Science of Information (CSoI), an NSF
Science and Technology Center.} \thanks{\IEEEauthorrefmark{1} Yeow-Khiang Chia was with Stanford University when most of this work was done. He is now with Institute for Infocomm Research, Singapore. Email: yeowkhiang@gmail.com}  \thanks{\IEEEauthorrefmark{2} Rajiv Soundararajan is with The University of Texas at Austin. Email: rajivs@utexas.edu}
\thanks{\IEEEauthorrefmark{3}Tsachy Weissman is with Stanford University. Email: tsachy@stanford.edu}%
}

\maketitle

\begin{abstract}
We consider the problem of estimating a signal corrupted by independent interference with the assistance of a cost-constrained helper who knows the interference causally or noncausally. When the interference is known causally, we characterize the minimum distortion incurred in estimating the desired signal. In the noncausal case, we present a general achievable scheme for discrete memoryless systems and novel lower bounds on the distortion for the binary and Gaussian settings. Our Gaussian setting coincides with that of assisted interference suppression introduced by Grover and Sahai. Our lower bound for this setting is based on the relation recently established by Verd\'{u} between divergence and minimum mean squared error. We illustrate with a few examples that  this lower bound can improve on those previously developed. Our bounds also allow us to characterize the optimal distortion in several interesting regimes. Moreover, we show that causal and noncausal estimation are not equivalent for this problem. Finally, we consider the case where the desired signal is also available at the helper. We develop new lower bounds for this setting that improve on those previously developed, and characterize the optimal distortion up to a constant multiplicative factor for some regimes of interest.  
\end{abstract}
\section{Introduction}
Consider a joint source channel coding problem as depicted in Figure~\ref{fig1}. We have two memoryless sources $S_1$ (the desired signal) and $S_2$ (the interfering signal). The decoder's aim is to estimate the source sequence $S_1^n$ from $Y^n$, with the goal of minimizing the average per symbol distortion $\E(\sum_{i=1}^nd(\Sh_{1i}(Y^n), S_{1i}))/n$. The encoder (helper), who knows the interfering signal $S_2$, aids the decoder in reconstructing the signal $S_1$ through his choice of $X$, subject to a cost constraint $\rho(X)$. 

\begin{figure}[!ht]
\begin{center}
\psfrag{enc}[c]{Enc.}
\psfrag{ch}[c]{$\P_{Y|X,S_1,S_2}$}
\psfrag{s1s2}[c]{$S_1, S_2$}
\psfrag{s2n}[l]{$S_2^n$}
\psfrag{s1n}[l]{$S_1^n$}
\psfrag{x}[c]{$X^n$}
\psfrag{y}[c]{$Y^n$}
\psfrag{dec}[c]{Dec.}
\psfrag{sh}[c]{$\Sh_1^n$}
\includegraphics[width = 0.9\linewidth]{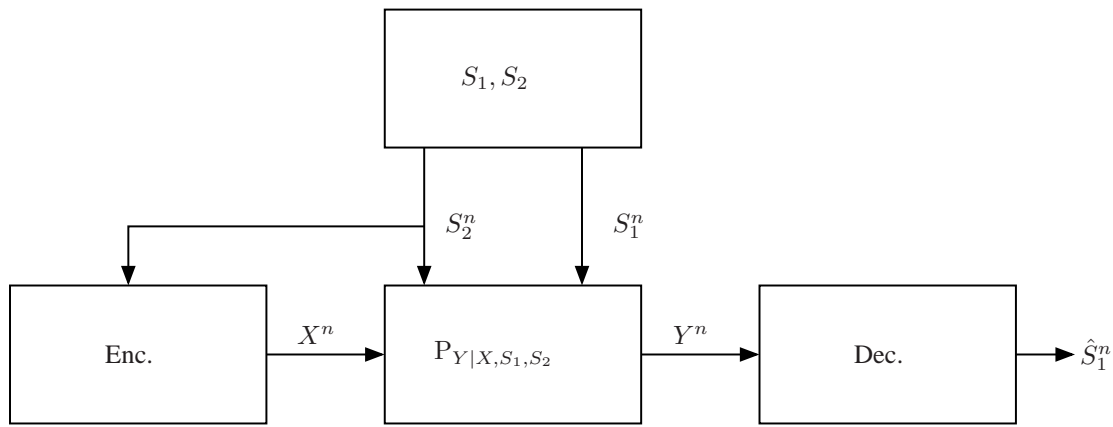}
\caption{Estimation with a helper who knows the interference. The interfering signal is $S_2^n$ while the desired signal is $S_1^n$. The encoder (helper) tries to help the decoder in estimating $S_1^n$ by reducing the interference due to $S_2$, subject to a per symbol cost constraint on its transmission $X^n$.} \label{fig1}
\end{center}
\end{figure}

Applications may arise in sensor networks or cognitive radio systems. As a motivating example, suppose Alice is talking to Bob in his office. As a result of ongoing construction work near Bob's office, there is high interference which makes it hard for Bob to listen to Alice. Fortunately, Bob recently purchased a noise cancellation device which has a microphone placed near the construction site. The microphone measures the interfering signal from the construction site and transmits it to a noise cancellation speaker situated in Bob's office. Since electromagnetic waves travel faster than sound, the noise cancellation speaker knows the interfering signal noncausally. Due to a power constraint on the speaker, it cannot cancel the interference fully. What then, is the minimum distortion that can be achieved by Bob in trying to reconstruct Alice's speech?

Our setup is closely related to several strands of work involving communication over channels with states. In \cite{Kim--Sutivong--Cover2008}, the authors considered the problem of State Amplification, where a message is to be sent to the decoder and the decoder also forms a list of possible $S_2^n$ sequences. The goal is to maximize the message transmission rate and reduce the uncertainty the decoder has regarding $S_2$; i.e. reduce the list size of possible $S_2^n$ sequences. Recently, the problem of state amplification with a distortion constraint was considered in \cite{Choudhuri--Kim--Mitra2011}, with an additional condition that the encoder only knows the state $S_2$ causally. This setting is similar to our setting, with the main difference being that the decoder wishes to reconstruct $S_2$ rather than $S_1$. When our setting is specialized to the Gaussian case with the mean squared error distortion between the reconstruction and the signal, our setting becomes equivalent to the problem of Assisted Interference Suppression considered in \cite{Grover--Sahai2011}. As detailed in \cite{Grover--Sahai2011}, this problem is closely related to Witsenhausen's counterexample in Optimum Control Theory \cite{Witsenhausen1968}.

In this paper, we consider both the case when $S_2$ is available causally at the encoder, and the case when $S_2$ is available non causally at the encoder. Our main contributions are as follows:

\begin{enumerate}
\item When $S_2$ is available causally at the encoder, we characterize the minimum achievable distortion in $S_1$. We borrow certain ideas used in the characterization of the distortion cost region for the causal state amplification problem in \cite{Choudhuri--Kim--Mitra2011} to establish our result. 

\item For the noncausal setting, we first give an achievable scheme for the general discrete memoryless system and then focus our attention on the case where $S_1$ and $S_2$ are independent Bernoulli random variables and the distortion measure is Hamming. We give two lower bounds on the achievable distortion for this binary setting. The first lower bound is based on ideas from the Assisted Interference Suppression problem \cite{Grover--Sahai2011}, while the second lower bound is based on ideas from the problem of Compression with Actions \cite{Lei--Chia--Weissman2011}. Neither bound contains the other and one bound can be better than the other, depending on the regime of interest. Using our lower and upper bounds, we characterize the minimum achievable distortion in several regimes. In particular, we provide an example to show that causal and noncausal estimation of $S_1$ are not equivalent and causal knowledge of $S_2$ could incur a higher distortion than noncausal knowledge of $S_2$ at the encoder. A complete characterization of the minimum achievable distortion in the noncausal case remains open.  

\item In the Gaussian case, where $S_1$ and $S_2$ are independent Gaussian random variables with finite variance, the distortion measure is the mean square error and $Y=X+S_1+S_2$, we note that our setting  coincides with that of Assisted Interference Suppression \cite{Grover--Sahai2011}. For this setting, we give a lower bound on the minimum achievable distortion which in some places improves on that given in \cite{Grover--Sahai2011}, and also its improved version given in \cite{Grover--Wagner--Sahai2011}. The proof of our lower bound relies on an application of Verdu's relation between relative entropy and mismatched estimation in Gaussian noise \cite{Verdu2010}. In recent years, since the seminal paper \cite{Guo--Shamai--Verdu2005} established the relationship between minimum mean square error estimation (MMSE) in Gaussian noise and the Mutual Information between the signal and the output, there has been interest in applying these information-estimation relations to problems in Information Theory (see e.g.~\cite{Shamai2011a} and~\cite{Guo--Wu--Shamai--Verdu2011}). Our lower bound, which seems difficult to obtain by traditional techniques such as the Entropy Power Inequality \cite[Chapter 2]{El-Gamal--Kim2010}, provides another application of these information-estimation relations.

\item In the Gaussian case, we also consider the setting when the encoder has access to $S_1$ noncausally, in addition to $S_2$. This setting is a special case of a problem considered in \cite{Huang--Narayanan2011}. We give a lower bound for this setting that contains the previous bounds in \cite{Huang--Narayanan2011} and can be strictly better in some cases. Furthermore, we establish constant gap results between the achievable distortion and our lower bound. 
\end{enumerate}

We first provide the formal definitions in the next section. In Section~\ref{sect:3}, we consider the causal case.  In Section~\ref{sect:4}, we consider the noncausal case, present an achievable scheme for general discrete memoryless systems and analyze the binary setting in detail.  Section~\ref{sect:5} deals with the Gaussian version of this problem, while we consider the Gaussian setting when $S_1$ is also available noncausally at the encoder in Section \ref{sect:6}. We conclude in Section VII with a summary of our findings and directions for future work. 

\section{Definitions} \label{sect:2}
In this section, we give formal definitions for our problem settings. We will follow the notation of \cite{El-Gamal--Kim2010}, and assume throughout this paper that the channel in consideration is memoryless. That is, $p(y^n|x^n, s_1^n, s_2^n) = \prod_{i=1}^n p(y_i|x_i, s_{1i}, s_{2i})$. We also assume that $S_1^n$ and $S_2^n$ are independent i.i.d. sequences. 
\subsection{Estimation with interference known at the helper} \label{sect:2_1}
A $(n, C)$ code for the setting shown in Figure \ref{fig1} when the interference is known \textit{noncausally} consists of
\begin{itemize}
\item An encoder that maps the interference $S_2^n$ to $X^n$, $f: \Sc_2^n \to \Xc^n$;
\item A decoder that maps the output $Y^n$ to the reconstruction sequence $\Sh_1^n$, $g: \Yc^n \to \hat{\Sc}^n_1$;
\end{itemize}
such that $\E \sum_{i=1}^n \rho(X_i)/n \le C$. The expected per symbol \textit{distortion}, $D$, is given by $D = \E d(S_1^n, \Sh_1^n) = \E\sum_{i=1}^n d(S_{1i}, \Sh_{1i})/n$.

A distortion $D$ is said to be achievable under the cost constraint $C$ if there exists a sequence of $(n, C + \e_n)$ codes, where $\e_n \to 0$ as $n \to \infty$, and 
\begin{align*}
\limsup_{n\to \infty} \E d(S_1^n, \Sh_1^n) \le D.
\end{align*}

The \textit{minimum achievable distortion}, $D(C)_{\rm min}$, is then defined as the infinum of the set of achievable distortions under the cost constraint $C$.   \\

When the interference is only known \textit{causally}, the definitions are mostly the same, with the difference being that the encoder is restricted to causal mapping:
\begin{align*}
f_i: \Sc_2^{i} \to \Xc \mbox{ for } i \in [1:n].
\end{align*}

\subsection{Estimation with source and interference known at the helper} \label{sect:svprob}
This setting is shown in Figure \ref{fig2}. For this setting, we restrict attention to the case where $S_1$ and $S_2$ are independent Gaussian random variables, $S_1 \sim \Nc(0, P_1)$ and $S_2 \sim \Nc(0, P_2)$. Furthermore, we assume that both $S_1$ and $S_2$ are known noncausally at the encoder, and the distortion measure is the mean square error between $S_1$ and its reconstruction. That is, $d(s_1, \sh_1) = (s_1- \sh_1)^2$. The channel is specified by $Y = X + S_1 + S_2 + Z$, where $Z \sim \Nc(0, N)$ is independent of $S_1$ and $S_2$. The cost constraint is the expected power constraint: $\E(\sum_{i=1}^nX_i^2/n)$. As the definitions are similar to the previous setting, we only mention the difference. That is, the encoder now maps both $S_1^n$ and $S_2^n$ to $X^n$: 
\begin{align*}
f: \Sc_1^n \times \Sc_2^n \to \Xc^n.
\end{align*}
\begin{figure}[!ht]
\begin{center}
\small
\psfrag{enc}[c]{Enc.}
\psfrag{ch}[l]{$Z$}
\psfrag{s1s2}[c]{$S_1, S_2$}
\psfrag{s2n}[l]{$S_2$}
\psfrag{s1n}[l]{$S_1$}
\psfrag{x}[c]{$X$}
\psfrag{y}[c]{$Y$}
\psfrag{dec}[c]{Dec.}
\psfrag{sh}[c]{$\Sh_1$}
\includegraphics[width = 0.9\linewidth]{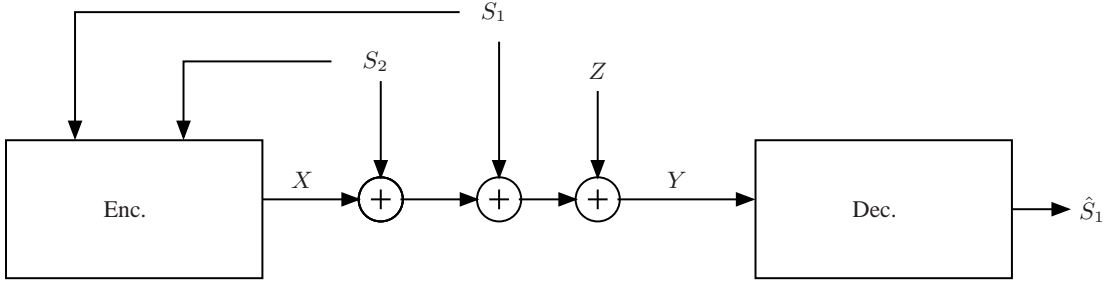}
\caption{Gaussian estimation with a helper that knows both the interference and the source. The random variables $S_1$, $S_2$ and $Z$ are independent zero mean Gaussian random variables. The encoder has knowledge of $S_1^n$ and $S_2^n$ noncausally and the decoder tries to perform lossy reconstruction of $S_1^n$. The distortion criterion is the mean square error criterion and the cost constraint is the expected power constraint on the encoder output, $X$.} \label{fig2}
\end{center}
\end{figure}

\section{Causal Estimation with a helper} \label{sect:3}
In this section, we give the distortion-cost tradeoff region for the setting given in \ref{sect:2_1} under the condition that the interfering signal, $S_2$, is causally known at the encoder. We will discuss some connections between our setting and that of the problem of Causal State Amplification discussed in \cite{Choudhuri--Kim--Mitra2011}. 

\begin{theorem} \label{thm:1}
The distortion-cost region for the problem of estimation with a helper when the interfering signal is \textit{causally} known at the encoder is given by
\begin{align*}
D(C)_{\rm min} =\min_{U,V,X, \Sh_1} \E d(S_1, \Sh_1(U,V,Y)) 
\end{align*}
for some $p(u)p(v|u,s_2)p(s_1)p(s_2)$ and functions $x(u,s_2)$ and $\sh_1(u,v,y)$ such that
\begin{align*}
I(U;Y) &\ge I(V;S_2|U,Y), \\
\E \rho(X) &\le C.
\end{align*}
The cardinalities of the auxiliary random variables may be upper bounded by $|\Uc| \le |\Sc_2|(|\Xc| - 1)+ 2$ and $|\Vc| \le |\Uc|(|\Sc_2|+1)$.
\end{theorem}

The achievability scheme in this Theorem is actually the same as that used in the problem of Causal State Amplification considered in \cite{Choudhuri--Kim--Mitra2011}, where the focus was on reconstructing $S_2$ instead of $S_1$. The expressions for the distortion-cost tradeoff are also similar, with the difference being that in the Causal State Amplification setting, one is interested in minimizing the distortion between $S_2$ and its reconstruction, rather than between $S_1$ and its reconstruction. Of course, the optimizing choice of auxiliary random variables in the two problems are different, since in our setting, we try to minimize the interference ($S_2$) as much as possible subjected to a cost constraint, whereas in the setting of Causal State Amplification, one tries to \textit{amplify} the interfering signal. As a (trivial) example, consider the case when $S_1, S_2, X \in \{0,1\}$ and $Y = X \oplus S_1 \oplus S_2$ and no cost constraint. Then, clearly, in our problem of causal estimation with a helper, we set $X = S_2$ to cancel out the interference completely, thereby recovering $S_1$ losslessly. In contrast, for the problem of Causal State Amplification, we will not cancel out $S_2$, since that is the signal we are trying to recover.

Theorem~\ref{thm:1} gives the optimal cost-distortion tradeoff for the estimation problem when the encoder knows the interfering signal causally. A natural question to ask is whether there is any penalty incurred in this restriction? In the next section, we will give an example of a binary estimation with a helper problem under Hamming loss and show that there is indeed a penalty incurred in only knowing the interfering signal causally.
 
\begin{proof}[Proof of Theorem \ref{thm:1}]

\textit{Sketch of Achievability:} As the achievability scheme is similar to that in \cite{Choudhuri--Kim--Mitra2011}, we give only a sketch in Appendix~\ref{appen_thm1} for completeness. 

\textit{Converse:} Given a $(n,C)$ code that achieves distortion $D$, we have{\allowdisplaybreaks
\begin{align*}
0 & \le \frac{1}{n}\sum_{i=1}^n I(Y^{n}_{i+1}; Y_i) \\
& = \frac{1}{n}\sum_{i=1}^n (I(Y_{i+1}^n, S_{2}^{i-1}; Y_i)  - I(S_{2}^{i-1}; Y_i | Y_{i+1}^n))\\
& \stackrel{(a)}{=} \frac{1}{n}\sum_{i=1}^n (I(Y_{i+1}^n, S_{2}^{i-1}; Y_i)  - I(Y_{i+1}^n; S_{2i} | S_{2}^{i-1})) \\
& \stackrel{(b)}{=} \frac{1}{n}\sum_{i=1}^n (I(Y_{i+1}^n, S_{2}^{i-1}; Y_i)  - I(Y_{i+1}^n, S_{2}^{i-1}; S_{2i} )) \\
& \stackrel{(c)}{=} I(Y_{Q+1}^n, S_{2}^{Q-1}; Y_Q|Q)  - I(Y_{Q+1}^n, S_{2}^{Q-1}; S_{2Q} |Q) \\
& \le  I(Y^{n}_{Q+1}, Q, S_{2}^{Q-1}; Y_Q)  - I(Y^{n}_{Q+1}, S_{2}^{Q-1}, Q; S_{2Q}) \\
& =  I(U, V; Y)  - I(U,V; S_2), 
\end{align*}}
where in $(a)$, we used the Csisz\'{a}r sum lemma \cite{Csiszar}; in $(b)$, we used the fact that $S_2$ is a memoryless source; in $(c)$, we defined $Q$ in the standard manner to be uniformly distributed over $[1:n]$ and independent of every other random variable; and in the last step, we define $U = (U_Q, Q) = (S^{Q-1}_2, Q)$ and $V = V_Q = Y_{Q+1}^n$. With these definitions of auxiliary random variables, it is clear that $U$ is independent of $S_2$ and also, the encoder output $X$, is a function of both $U$ and $S_2$. Further, using the relationship that $U$ is independent of $S_2$ and $V - (U,S_2) - Y$, the condition that $I(U, V; Y)  - I(U,V; S_2) \ge 0$ reduces to
\begin{align*}
I(U;Y) &\ge I(V;S_2|U,Y).
\end{align*}

It now remains to show that the achievable distortion can be lower bounded by this choice of auxiliary random variables. To this end, we will use a technique for lower bounding distortion found in \cite{Kaspi}. We have
\begin{align}
D+ \e_n &\ge \frac{1}{n}\sum_{i=1}^n \E d(S_{1i}, \Sh_{1i}(Y^n)) \nonumber\\
&= \frac{1}{n}\sum_{i=1}^n \E d(S_{1i}, \Sh_{1i}(Y_i, Y_{i+1}^n, Y^{i-1})) \nonumber\\
&= \frac{1}{n}\sum_{i=1}^n \E d(S_{1i}, \Sh_{1i}(Y_i, V_i, Y^{i-1})) \nonumber \\
& \ge \frac{1}{n}\sum_{i=1}^n \E d(S_{1i}, \Sh_{1i}'(Y_i, V_i, Y^{i-1}, S_{2}^{i-1})), \label{ineq:1}
\end{align}
where the last step follows from the observation that we can recover $\Sh_{1i}$ from $\Sh_{1i}'$ by simply ignoring $S_{2}^{i-1}$. Next, consider the term $\E d(S_{1i}, \Sh_{1i}'(Y_i, V_i, Y^{i-1}, S_{2}^{i-1}))$.{\allowdisplaybreaks
\begin{align}
&\E d(S_{1i}, \Sh_{1i}'(Y_i, V_i, Y^{i-1}, S_{2}^{i-1})) \nonumber \\
 &= \E d(S_{1i}, \Sh_{1i}'(Y_i, V_i, Y^{i-1}, U_i))\nonumber\\
&=  \sum p(s_{1i}, u_i, v_i, y_i, y^{i-1}) d(s_{1i}, \Sh_{1i}'(y_i, v_i, y^{i-1}, u_i)) \nonumber\\
& = \sum p(u_{i}, y_{i}, v_i) \sum p(y^{i-1}, s_{1i}| u_{i}, y_{i}, v_i)d(s_{1i}, \Sh_{1i}'(y_i, v_i, y^{i-1}, u_i))\nonumber \\
& \stackrel{(a)}{= } \sum p(w_i) \sum p(y^{i-1}|w_i)p( s_{1i}|w_i)d(s_{1i}, \Sh_{1i}'(w_i, y^{i-1}))\nonumber\\
& =  \sum p(w_i) \sum_{y^{i-1}} p(y^{i-1}|w_i)\sum_{s_{1i}}p( s_{1i}|w_i)d(s_{1i}, \Sh_{1i}'(w_i, y^{i-1}))\nonumber\\
& \stackrel{(b)}{\ge} \sum p(w_i) \sum_{y^{i-1}} p(y^{i-1}|w_i)\sum_{s_{1i}}p( s_{1i}|w_i)d(s_{1i}, \Sh_{1i}^*(w_i))\nonumber\\
& = \sum p(w_i, s_{1i})d(s_{1i}, \Sh_{1i}^*(w_i)) \nonumber\\
& = \E d(S_{1i}, \Sh_{1i}^*(W_i)), \label{ineq:2}
\end{align} }
where in $(a)$, we define $w_i = (u_{i}, y_{i}, v_i)$ for notational convenience and the fact that $p(y^{i-1}, s_{1i}| w_i) = p(y^{i-1}|w_i)p(s_{1i}|w_i)$ follows from the Markov Chain $Y^{i-1} - W_i - S_{1i}$, which in turn, follows from the fact that $S_{2}$ is only causally known at the encoder. Hence, given $S_2^{i-1}$ and also $X^{i-1}$ since it is a function of $S_2^{i-1}$, $Y^{i-1}$ is independent of $S_{1i}$. $(b)$ follows from defining $y^{i-1*} = \arg\min_{y^{i-1}} \sum_{s_{1i}}p( s_{1i}|w_i)d(s_{1i}, \Sh_{1i}'(w_i, y^{i-1}))$ and $\Sh_{1i}^*(w_i) = \Sh_{1i}'(w_i, y^{i-1*})$. 

Combining inequality \eqref{ineq:2} into inequality \eqref{ineq:1} then gives us
\begin{align*}
D+ \e_n &\ge \frac{1}{n}\sum_{i=1}^n \E d(S_{1i}, \Sh_{1i}^*(Y_i, V_i, U_i)) \\
& = \E_Q (\E(d(S_{1Q}, \Sh_{1Q}^*(Y_Q, V_Q, U_Q))|Q)) \\
& \ge \E(d(S_1, \Sh_1(Y, V, U)). 
\end{align*}
The bounds on cardinality of the auxiliary random variables follow from standard arguments (see for e.g.~\cite[Appendix C]{El-Gamal--Kim2010}). This completes the proof of converse. 
\end{proof}

\section{Noncausal Estimation with a helper} \label{sect:4}
Having established the distortion-cost region for the discrete memoryless estimation with a helper problem when the interfering signal is causally known, we now turn to the noncausal setting, that is, when $S_2$ is noncausally known at the encoder. This setting is more complicated and the distortion-cost region is still unknown. In this section, we first give an achievability scheme based on the recently proposed technique of hybrid coding \cite{Lim--Minero--Kim2010}. We then specialize our setting to the case of binary estimation with a helper. 

The problem of binary estimation with a helper is one where $S_1 \sim \Bern(p_1)$, $S_2 \sim \Bern(p_2)$, $0 \le p_1, p_2 \le 1/2$, $X \in \{0,1\}$, $Y = X \oplus S_1 \oplus S_2$ and $d(S_1, \Sh_1) = S_1 \oplus \Sh_1$, i.e., Hamming distortion. The cost is given by $\rho(X) = 1$ if $X = 1$ and $0$ otherwise. The objective of the problem is to design a coding strategy that minimizes the Hamming distortion in $S_1$. 

Specializing to the case of binary estimation with a helper allows us to derive a number of additional results of interest. In subsection~\ref{sect:4_1}, we give a   (non-trivial) condition on the cost constraint that allows us to achieve zero expected distortion. We then show that  in the binary case, there is a penalty involved if $S_2$ is known only causally instead of noncausally. As a result, the distortion incurred in $S_1$ is higher if $S_2$ is only known causally as opposed to it being known noncausally. In subsection~\ref{sect:4_2}, we describe the two lower bounds for the problem of binary estimation with a helper and then compare them. In subsection~\ref{sect:4_3}, we briefly mention a non-binary setting for which we can characterize the distortion-cost tradeoff, and show that symbol by symbol encoding is optimal in that setting. 
  
\subsection{Achievable scheme} \label{sect:4_1}
We first give an achievable scheme for the general discrete memoryless estimation with a helper problem based on hybrid coding~\cite{Lim--Minero--Kim2010}. We will extend this scheme to the Gaussian case in the next section. 

\begin{theorem} \label{thm:2}
An achievable distortion for the problem of estimation with a helper is given by
\begin{align*}
D(C) \le \inf \E d(S_1, \Sh_1(U,Y)),
\end{align*}
where the minimization is over distribution $p(u|s_2)$ and functions $x = f(s_2, u)$ and $\sh_1(u,y)$ such that
\begin{align*}
I(U;Y) &> I(U;S_2), \\
\E \rho(X)& \le C.
\end{align*}
\end{theorem}

\textit{Sketch of Achievability:} The achievability scheme follows that of the hybrid coding scheme given in~\cite{Lim--Minero--Kim2010}. We give only a sketch here. The codebook generation consists of generating $2^{n(I(U;S_2) + \e)}$ sequences according to $\prod_{i=1}^n p(u_i)$. For encoding, given an $s_2^n$ sequence, the encoder looks for a $u^n$ sequence such that $(u^n, s_2^n) \in \aep$. If there is more than one, it selects one sequence uniformly at random from the set of jointly typical sequences. It then outputs $x^n$ according to $f(u_i, s_{2i})$ for $i \in [1:n]$. The decoder looks for the unique $\uh^n$ sequence such that $(\uh^n, y^n) \in \aep$. It can be shown as in ~\cite{Lim--Minero--Kim2010} that the probability of decoding error goes to zero as $n \to \infty$ if
\begin{align*}
I(U;Y) > I(U;S_2) + 2\e. 
\end{align*} 
The decoder then reconstructs $S_1^n$ according to $\sh_1(\uh_i, y_i)$ for $i \in [1:n]$. 

We now specialize the achievable distortion-cost region in Theorem~\ref{thm:2} to the case of binary estimation with a helper. The next result shows that, in the binary case, zero expected distortion is achievable under a condition on the cost constraint.

\begin{proposition} \label{prop1}
For the problem of binary estimation with a helper,
\begin{align*}
D(C)_{\rm min} = 0
\end{align*}
if $H_2(C) > H(X\oplus S_2 |Y)$, where $H_2(.)$ is the binary entropy function, $X\sim \Bern(C)$ independent of $S_2$ and $Y = X\oplus S_1 \oplus S_2$. 
\end{proposition}
\begin{proof}
The sufficient condition on the cost constraint follows from a particular choice of auxiliary random variable $U$ in Theorem \ref{thm:2}. We let $X \sim \Bern(C)$ independent of $S_2$ and let $U = X\oplus S_2$. The decoder reconstructs $S_1$ as $\Sh_1 = Y \oplus U = S_1$, incurring zero expected distortion. We now note that the cost constraint is satisfied since $X\sim \Bern(C)$. To satisfy the mutual information condition on the choice of joint distribution, we require
\begin{align*}
I(U;Y) &> I(U;S_2) \\
\Rightarrow H(U|S_2) &> H(U|Y) \\
\Rightarrow H(X|S_2) &> H(X\oplus S_2|Y) \\
\Rightarrow H_2(C) &> H(X\oplus S_2|Y).
\end{align*}
\end{proof}

Weakening Proposition~\ref{prop1} leads to the following simple sufficient condition for zero distortion.
\begin{corollary} \label{coro1}
If $C > p_1$, $D(C)_{\rm min} = 0$.
\end{corollary}

Proof of Corollary~\ref{coro1} follows readily from Proposition \ref{prop1}. Since $0\le C,p_1 \le 1/2$, if $C> p_1$, then
\begin{align*}
H_2(C) &> H_2(p_1) \\
& = H(S_1)\\
& \ge H(S_1|Y) \\
& = H(X\oplus S_2 |Y).
\end{align*}

\begin{remark}
A trivial condition for zero distortion is when $C \ge p_2$ in which case, the encoder just performs symbol by symbol cancellation of $S_2$ to allow the decoder to recover $S_1$ losslessly. Corollary \ref{coro1} shows that zero expected distortion can be achieved even if $C < p_2$ as long as $C > p_1$. 
\end{remark}

By choosing $U = (U',V')$ in Theorem \ref{thm:2}, where $p(u|s_2) = p(u')p(v'|u',s_2)$, we obtain the distortion-cost region when $S_2$ is restricted to be causally known at the encoder\footnote{The boundary case of $I(U;Y) = I(V;S_2|U,Y)$ is treated in a similar fashion as in the causal setting.}. A natural question to ask is whether the achievable distortion for the same cost constraint can be lowered if $S_2$ is noncausally known at the encoder rather than only causally known. This is indeed the case for the problem of binary estimation with a helper.

\begin{proposition} \label{prop2}
For the problem of binary estimation with a helper, the achievable distortion when $S_2$ is noncausally known at the encoder can be \textit{strictly} smaller than the achievable distortion when $S_2$ is only causally known at the encoder, with the same cost constraint.
\end{proposition} 
\begin{proof}
To prove Proposition~\ref{prop2}, we exhibit an example where we can achieve zero expected distortion when $S_2$ is noncausally known at the encoder, but for which the achievable distortion is strictly greater than zero when $S_2$ is only causally known. To this end, we let $p_1 = 0.1$, $p_2 = 0.5$ and $C = 0.11$. Since $C > p_1$, from Corollary~\ref{coro1}, an expected distortion of $0$ can be achieved when $S_2$ is noncausally known at the encoder. That is, we have $D(0.11)_{\rm min-noncausal} = 0$. Proof of this proposition is completed using the following claim, which states that the minimum expected distortion when $S_2$ is only causally known at the encoder, $D(0.11)_{\rm min-causal}$, is \textit{strictly} greater than zero. 

\begin{claim} \label{clm1}
$D(0.11)_{\rm min-causal} >0$ for any choice of $U,V$ satisfying the constraints given in Theorem~\ref{thm:1}.
\end{claim}
Claim~\ref{clm1} is proven in Appendix~\ref{appen_claim}.
\end{proof}

\subsection{Lower bounds for binary estimation with helper} \label{sect:4_2}
We now turn to lower bounds for the binary estimation with a helper problem. The first lower bound that we will present uses ideas from~\cite{Grover--Sahai2011} adapted from the Gaussian to the binary setting. 

\begin{theorem} \label{thm:3}
A lower bound for the achievable distortion for the problem of binary estimation with a helper is given by
\begin{align*}
D(C)_{\rm min} \ge \min H_2^{-1}(H(S_1) + H(S_2) - H(Y)) - \E X,
\end{align*}
where we define $H_2^{-1} = 0$ if the argument is negative or greater than 1, and the minimization is over joint distribution $p(x|s_2)$ such that $\E X \le C$.
\end{theorem}
\begin{proof}
We first start with a simple claim.
\begin{claim} \label{clm2}
Let $\sh_1^n(y^n)$ be an optimal reconstruction function (with respect to Hamming distortion) for $s_1^n$ and $\xh^n(y^n)$ be an optimal reconstruction for $s_2^n\oplus x^n$. Then, $d(s_1^n, \sh_1(y^n)) = d(s_2^n \oplus x^n, \xh^n(y^n))$.
\end{claim}{\allowdisplaybreaks
To prove this claim, observe that $d(s_1^n, \sh_1(y^n)) = \sum_{i=1}^n s_{1i} \oplus \sh_{1i}(y^n)$. Consider now the function $\xh'_{i}(y^n) = \sh_{1i}(y^n) \oplus y_i$. Since $\xh_i(y^n)$ is optimal for $d(s_2^n \oplus x^n, \xh^n(y^n))$, we have
\begin{align*}
d(s_2^n \oplus x^n, \xh^n(y^n)) &\le \sum_{i=1}^n s_{2i} \oplus x_i\oplus \xh'_i(y^n) \\
& =  \sum_{i=1}^n s_{2i} \oplus x_i\oplus \sh_{1i}(y^n)\oplus y_i \\
& = \sum_{i=1}^n s_{1i} \oplus \sh_{1i}(y^n) \\
& = d(s_1^n, \sh_1(y^n)).
\end{align*} }

Hence, we have $d(s_2^n \oplus x^n, \xh^n(y^n)) \le d(s_1^n, \sh_1(y^n))$. For the other direction, consider the function $\sh'_{1i} = \xh_{i}(y^n) \oplus y_i$. Repeating the same arguments for $d(s_1^n, \sh_1(y^n))$ instead of $d(s_2^n \oplus x^n, \xh^n(y^n))$, it is easy to show that $d(s_1^n, \sh_1(y^n)) \le d(s_2^n \oplus x^n, \xh^n(y^n))$, which completes the proof of claim~\ref{clm2}. 

As an aside, the proof of claim~\ref{clm2} shows that the optimal reconstruction functions for the respective problems are related by $\xh^n(y^n) = \sh_1^n(y^n) \oplus y^n$.

We now continue with our lower bound for the binary case. Using claim~\ref{clm2}, we have
\begin{align}
d(S_1^n, \sh_1^n(Y^n)) &= d(X^n\oplus S_2^n, \xh^n(Y^n)) \nonumber \\
& \ge d(S_2^n, \xh^n(Y^n))- d(S_2^n, X^n\oplus S_2^n). \label{ineq:5}
\end{align}
The second line follows from the fact that the Hamming distance is a proper distance metric, and it therefore satisfies the triangular inequality. Hence,
\begin{align*}
\frac{1}{n}\E d(S_1^n, \sh_1^n(Y^n)) &= \frac{1}{n}\E d(X^n\oplus S_2^n, \xh^n(Y^n)) \\
& \ge \frac{1}{n}\E d(S_2^n, \xh^n(Y^n))- \frac{1}{n}\E d(S_2^n, X^n\oplus S_2^n).
\end{align*}
Let $Q$ be uniform $[1:n]$, independent of other random variables. Then,
\begin{align}
\frac{1}{n}\E d(S_2^n, X^n\oplus S_2^n) &= \E( \frac{1}{n} \sum_{i=1}^n X_i) \nonumber \\
& = \E X_Q \nonumber \\
& = \E X \label{ineq:6} 
\end{align}
This is the expected number of ones in $X^n$. For the term, $\E d(S_2^n, \xh^n(Y^n))/n$, we lower bound it by
\begin{align}
\frac{1}{n}\E d(S_2^n, \xh^n(Y^n)) &\ge \frac{1}{n}\sum_{i=1}^n\E d(S_{2i}, \sh_{2i}(Y^n)) \label{ineq:3}
\end{align}
where $\sh_2(Y^n)$ is an optimal reconstruction function with respect to Hamming distortion for $S_2$. The right hand side of inequality~\eqref{ineq:3} is then further lower bounded by the following argument. From data processing inequality~\cite{Cover}, we have

\begin{align*}
I(S^n_2; \Sh_2^n) &\le I(S_2^n;Y^n) \\
& \le \sum_{i=1}^n (H(Y_i) -H(Y_i|S_2^n, X^n)) \\
& =  \sum_{i=1}^n H(Y_i) -nH(S_{1i}) \\
& \le nH(Y) - nH(S_1). 
\end{align*}
On the other hand,
\begin{align*}
I(S^n_2; \Sh_2^n) &\ge \sum_{i=1}^n (H(S_{2i}) - H(S_{2i} \oplus \Sh_{2i})) \\
& \ge nH_2(S_2) - nH_2\left(\frac{1}{n}\sum_{i=1}^n\E d(S_{2i}, \sh_{2i}(Y^n))\right),
\end{align*}
where the last line follows from concavity of entropy~\cite{Cover}.
Combining the upper and lower bounds gives us
\begin{align}
\frac{1}{n}\sum_{i=1}^n\E d(S_{2i}, \sh_{2i}(Y^n)) \ge H_2^{-1}(H(S_1) + H(S_2) - H(Y)), \label{ineq:4}
\end{align}
where we define $H_2^{-1}(.) :=0$ if the argument is negative or greater than 1.

Substituting \eqref{ineq:3}, \eqref{ineq:4} and \eqref{ineq:6} into \eqref{ineq:5}, we have
\begin{align*}
D(C)_{\rm min} \ge H_2^{-1}(H(S_1) + H(S_2) - H(Y)) - \E X,
\end{align*}
where $\E X \le C$ from the cost constraint. 
\end{proof}

Using the lower bound in Theorem~\ref{thm:3}, we can show that when $p_1 = 1/2$, symbol by symbol cancellation of $S_2$ is optimal and hence, when $p_1 = 1/2$, the minimum achievable distortion for the same cost constraint is the same regardless of whether $S_2$ is known causally  or noncausally.

\begin{proposition} \label{prop3}
When $p_1 = 1/2$ and $p_2 > C$, the distortion-cost region is given by
\begin{align*}
D(C)_{\rm min} = p_2 -C.
\end{align*}
\end{proposition}
\begin{proof}
When $S_1 \sim \Bern(1/2)$, $Y \sim\Bern(1/2)$, regardless of the distribution of $S_2 \oplus X$. Hence, Theorem~\ref{thm:3} reduces to 
\begin{align*}
D(C)_{\rm min} &\ge p_2 -\E X \\
& \ge p_2 - C.
\end{align*}
Achievability of this lower bound follows from Theorem~\ref{thm:1} by setting $V = \emptyset$, $U$ to be a random variable such that
\begin{align*}
X = \left\{\begin{array}{clc} 1 & \mbox{w. p. } \frac{C}{p_2} & \mbox{if } S_2=1\\
0 & \mbox{w. p. } 1-\frac{C}{p_2} & \mbox{if } S_2=1 \\
0 & \mbox{otherwise}
\end{array}\right. .
\end{align*}
The existence of such a $U$ follows from the functional representation lemma~\cite[Appendix B]{El-Gamal--Kim2010}. It is easy to verify that the expected cost constraint is satisfied with this choice of distribution $p(x|s_2)$. The reconstruction function in this case is simply $\Sh_1 = Y$. It also easy to verify that the distortion constraint is satisfied. 
\end{proof}

The optimization problem in Theorem~\ref{thm:3} can be simplified in a number of cases.
\begin{corollary} \label{coro2}
Theorem~\ref{thm:3} simplifies under the following conditions

\begin{enumerate}
\item Under the condition $p_1 +(1-2p_1)(p_2 -C) \ge 1/2$, Theorem~\ref{thm:3} simplifies to
\begin{align*}
D(C)_{\rm min} \ge H_2^{-1}(H(S_1) + H(S_2) - H(p_1 +(1-2p_1)(p_2 -C))) - C.
\end{align*}
\item Under the condition $p_1 +(1-2p_1)(p_2 +C)\le 1/2$, Theorem~\ref{thm:3} simplifies to
\begin{align*}
D(C)_{\rm min} \ge H_2^{-1}(H(S_1) + H(S_2) - H(p_1 +(1-2p_1)(p_2 -C))) - C.
\end{align*}
\end{enumerate}
\end{corollary}
\begin{proof}
The proof follows from observing that $p_2 -C \le \E X\oplus S_2 \le p_2 + C$. Define $\E X \oplus S_2 : = p_{x \oplus s_2}$. Then, $Y \sim \Bern( p_1 + (1- 2p_1)p_{x \oplus s_2})$. If condition one in the corollary is satisfied, then $H(Y)$ is a decreasing function of $p_{x\oplus s_2}$. It is then easy to see from the expression in Theorem~\ref{thm:3} that the minimizing distribution is one where $p_{x\oplus s_2} = p_2 - C$ and $\E X = C$. A similar proof applies for the second condition, which completes proof of this corollary. 
\end{proof} 
It appears to be quite difficult to obtain an explicit analytical solution for the general case of $p_1 +(1-2p_1)(p_2 -C)< 1/2< p_1 +(1-2p_1)(p_2 +C)$. A looser bound in this case is
\begin{corollary}
\begin{align*}
D(C)_{\rm min} \ge H_2^{-1}(H(S_1) + H(S_2) - 1) - C.
\end{align*}
\end{corollary}
Proof of this corollary is omitted as it follows directly from Theorem~\ref{thm:3}.

We now present another lower bound for the binary setting, using ideas from the proof of converse for Gel'fand-Pinsker coding given in~\cite[Chapter 7]{El-Gamal--Kim2010}, and also ideas from~\cite{Lei--Chia--Weissman2011}. The main intuition in this lower bound comes from Claim~\ref{clm2} used in the proof of Theorem~\ref{thm:3}, which shows that the optimum distortion incurred in reconstructing $X\oplus S_2$ is the same as the optimum distortion incurred in reconstructing $S_1$. We then try to lower bound $D(C)_{\rm min}$ by lower bounding the distortion incurred in reconstructing $X \oplus S_2$. We will see in the sequel that in some cases, this lower bound is better than the previous lower bound given in Theorem~\ref{thm:3}.
\begin{theorem} \label{thm:4}
A lower bound for the achievable distortion for the problem of binary estimation with a helper is given by
\begin{align*}
D(C)_{\rm min} \ge \min H_2^{-1}(H(S_1) + H(X\oplus S_2|U) + I(U;S_2) - H(Y)),
\end{align*}
where we minimize over $p(u|s_2)$ and $x = f(u,s_2)$ such that $\E X \le C$. The cardinality of the auxiliary random variable $U$ may be upper bounded by $|\Uc| \le |\Sc_2|(|\Xc|-1)+2$. In the binary case that we are interested in, $|\Uc| \le 4$.
\end{theorem}
\begin{proof}
For notational convenience, let $\Zh$ represent the optimal reconstruction for $X \oplus S_2$ and $Z = X \oplus S_2$. From data processing inequality,
\begin{align*}
I(Z^n; \Zh^n) \le I(Y^n; Z^n).
\end{align*}
On the one hand,{\allowdisplaybreaks
\begin{align*}
I(Z^n; \Zh^n) &\ge H(Z^n) - H(\Zh^n | Z^n) \\
&\ge \sum_{i=1}^n (H(Z_i|Z^{i-1}) -H(\Zh_i \oplus Z_i) ) \\
& \stackrel{(a)}{\ge} \sum_{i=1}^n (H(Z_i|Z^{i-1}, S_{2,i+1}^n) + I(S_{2,i+1}^n; Z_i |Z^{i-1}) )- H_2(D(C)_{\rm \min}) \\
& = \sum_{i=1}^n (H(Z_i|Z^{i-1}, S_{2,i+1}^n) + I(Z^{i-1}; S_{2i} | S_{2,i+1}^n) )- H_2(D(C)_{\rm \min}) \\
& = \sum_{i=1}^n (H(Z_i|Z^{i-1}, S_{2,i+1}^n) + I(Z^{i-1}, S_{2,i+1}^n; S_{2i}  ) )- H_2(D(C)_{\rm \min}) \\
& = \sum_{i=1}^n (H(Z_i|U_i) + I(U_i; S_{2i}  ) )- H_2(D(C)_{\rm \min}).
\end{align*}
In $(a)$, we used concavity of entropy and Claim~\ref{clm2}, which states that the optimum distortion for $X \oplus S_2$ is the same as the optimum distortion for $S_1$.
Next, 
\begin{align*}
I(Y^n;Z^n) &= H(Y^n) - H(Y^n|Z^n) \\
& \le \sum_{i=1}^n H(Y_i) - H(S_1^n) \\
& = \sum_{i=1}^n (H(Y_i) - H(S_{1i})).
\end{align*}
}
Defining the standard $Q$ uniform random variable over $[1:n]$ independent of other random variables, $U = (U_Q,Q)$, $Y_Q = Y$, $S_{1Q} = S_1$, $S_{2Q} = S_2$ and $Z_Q = Z$ then gives us the following lower bound
\begin{align*}
D \ge H_2^{-1}(H(S_1) + H(Z|U) + I(U;S_2) - H(Y)),
\end{align*}
where we minimize over $p(u|s_2)p(x|u,s_2)$ such that $\E X \le C$. Reducing the cost constraint to this single letter expression ($\E X \le C$) follows the same procedure as in Theorem~\ref{thm:3}.

Next, we note that instead of minimizing over $p(x|u,s_2)$, it suffices to minimize over $x = f(u,s_2)$. To see this, note that we can always find a $V$, independent of $U,S_2$, such that $p(x|u,s_2) = f(u,v,s_2)$. Now, define $\Ut = (U,V)$. Observe that since we preserve both $p(x\oplus s_2)$ and $p(x)$, the cost constraint and $H(Y) = H(Z \oplus S_1)$ remains unchanged. Now, note that 
\begin{align*}
H(Z|\Ut) \le H(Z|U),
\end{align*} 
and 
\begin{align*}
I(\Ut; S_2) &= I(U;S_2) + I(V;S_2|U) \\
& = I(U;S_2).
\end{align*}
The bound on the cardinality of $U$ follows from standard techniques and we omit it here. This completes the proof of the lower bound.
\end{proof}
Theorem~\ref{thm:4} involves minimizing over joint distributions and choice of auxiliary random variable $U$. A looser bound that is easier to compute is given by the following corollary.
\begin{corollary} \label{coro4}
\begin{align*}
D(P)_{\rm min} \ge H_2^{-1}(H_2(p_1) + \min_{p_2 -c \le \alpha \le p_2+c} \{H_2(\alpha) - H_2(\alpha*p_1)\} - I(U;Z) + I(U;S_2)).
\end{align*}
for some joint distribution $p(u|s_2)$ and $x = f(u,s)$ satisfying $\E X \le C$, and $Z = X\oplus S_2$. 
\end{corollary}
In Corollary~\ref{coro4}, we need to perform maximization of $I(U;Z) - I(U;S_2)$ subjected to a cost constraint $\E X \le C$. This is nothing but the problem of maximization of the capacity of a Gel'fand-Pinsker channel subjected to a cost constraint. There are efficient numerical algorithms for performing this maximization, cf.\ ~\cite[Page 555-556]{Keshet--Steinberg--Merhav2007} for a description of the algorithm. 
\begin{proof}

Starting from Theorem~\ref{thm:4}, consider the term $H(Z|U) + I(U;S_2) - H(Y)$ in the Theorem.
\begin{align*}
H(Z|U) + I(U;S_2) - H(Y) &= H(Z,U) - H(U|S_2) - H(Y) \\
& = H(Z) - H(Y) + H(U|Z) - H(U|S_2) \\
& = (H(Z) - H(Y)) -(I(U;Z) - I(U;S_2)).
\end{align*}

We now minimize the terms $(H(Z) - H(Y))$ and $-(I(U;Z) - I(U;S_2))$ separately. We have discussed maximizing the term $I(U;Z) - I(U;S_2)$ earlier. As for the term $(H(Z) - H(Y))$, using the observation $p_2 -C \le \E Z \le p_2 +C$, we have
\begin{align*}
\min \{H(Z) - H(Y)\} = \min_{p_2 -C \le \alpha \le p_2+C} \{H_2(\alpha) - H_2(\alpha*p_1)\},
\end{align*}
which completes the proof.
\end{proof}

\subsection*{Comparison of lower bounds}
As we mentioned, the expressions in Theorems~\ref{thm:3} and~\ref{thm:4} can be difficult to compute. For the purpose of simulations, we compare the expressions of Corollary~\ref{coro2} with those of Corollary~\ref{coro4}, when the conditions given in Corollary~\ref{coro2} are satisfied. Note that since Corollary~\ref{coro4} can be weaker than Theorem~\ref{thm:4} whereas Corollary~\ref{coro2} gives the same bounds as Theorem~\ref{thm:3} when the conditions are satisfied, an advantage of this comparison is that it shows when Theorem~\ref{thm:4} can be strictly larger than Theorem~\ref{thm:3}. 

For our numerical example, we set $p_2 = 0.1$, vary the cost from $0.01$ to $0.03$ and compute plots for $p_1 = 0.05, 0.09$. In general, the bound in Theorem~\ref{thm:3} is better, but we focus on small values of cost, $p_1$ and $p_2$ to show that there are regimes in which the expression in Theorem~\ref{thm:4} is better. The plots are shown in Figures~\ref{fig:bp1} and~\ref{fig:bp2}. As can be seen in Figure~\ref{fig:bp1}, there are regions for which Theorem~\ref{thm:4} is strictly better than Theorem~\ref{thm:3}. However, Theorem~\ref{thm:3} does give a better bound for a wider range of values as compared to Theorem~\ref{thm:4}. 

\begin{figure}[!h]
\psfrag{C2}[l]{Corollary~\ref{coro2}}
\psfrag{C4}[l]{Corollary~\ref{coro4}}

\begin{center}
\scalebox{0.75}{\includegraphics{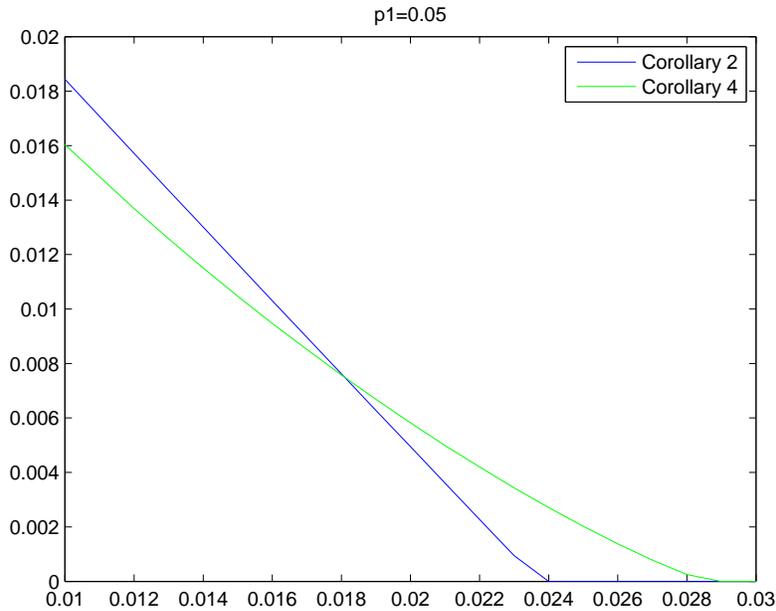}}
\caption{Comparison of bounds for $p_1 = 0.05$. Y-axis represents the distortion level while X-axis represents the cost.} \label{fig:bp1}
\end{center}
\end{figure}  

\begin{figure}[!h]
\psfrag{C2}[l]{Corollary~\ref{coro2}}
\psfrag{C4}[l]{Corollary~\ref{coro4}}

\begin{center}
\scalebox{0.75}{\includegraphics{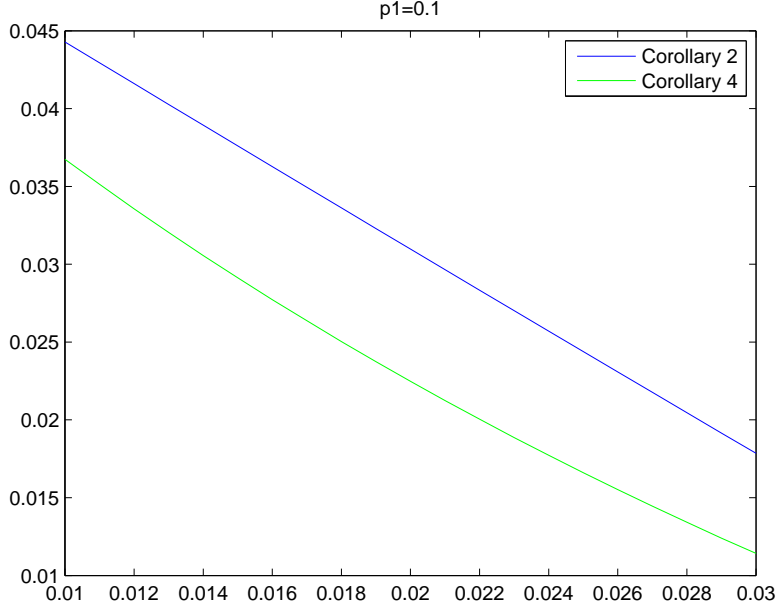}}
\caption{Comparison of bounds for $p_1 = 0.1$. Y-axis represents the distortion level while X-axis represents the cost. In this case, the bound given by Corollary~\ref{coro2} is strictly better than that for Corollary~\ref{coro4}.} \label{fig:bp2}
\end{center}
\end{figure}  

\subsection{Erasure estimation with helper} \label{sect:4_3}

For most of this section, we have focused on the binary estimation with helper setup. In this subsection, we briefly mention a setting, erasure estimation with helper, for which we can characterize the distortion-cost function and also, for which symbol by symbol cancellation of $S_2$ is optimal.

The setting is defined by $S_1 \sim p(s_1)$, $S_2 \sim \Bern(p_2)$, $X \in \{0,1\}$ and $Y$ is defined as follows
\begin{align*}
Y = \left\{\begin{array}{ccc} S_1 & \mbox{ if } & X\oplus S_2 =0 \\
e & \mbox{ if } & X\oplus S_2 = 1 \end{array} \right. .
\end{align*}

This is a model of a channel in which when the interfering signal is large, the desired signal is erased. When the interfering signal is small, decoder receives the signal perfectly. The helper tries to help the decoder by canceling the interference. The distortion-cost region is characterized by the following proposition.
\begin{proposition} \label{prop4}
The distortion-cost region for the problem of erasure estimation with helper is given by
\begin{align*}
D(C)_{\rm min} &= \min\left\{ \P(X \oplus S_2 = 1)(\min_{\sh_1} \E d(S_1, \sh_1)) \right.\\
& \left.\qquad \qquad + \P(X \oplus S_2 = 0) (\min_{\sh_1} \E d(S_1, \sh_1))\right\},
\end{align*}
where the minimization is over $p(x|s_2)$ satisfying $\E \rho(X) \le C$.
\end{proposition}
\begin{proof}
Achievability of the distortion-cost region uses a modified version of the achievability scheme used in Proposition~\ref{prop3}. The modification comes in the reconstruction function where
\begin{align*}
\sh_1(Y) = \left\{ \begin{array}{ll} \arg \min_x d(Y, x) & \mbox{ if } Y = S_1, \\
\arg \min_x \E d(S_1, x) & \mbox{ if } Y = e \end{array}\right. .
\end{align*}
With this choice of reconstruction function and noting that $\P(Y = S_1) = \P(X\oplus S_1 = 0)$ and $\P(Y = e) = \P(X\oplus S_1 = 1)$, it is easy to see that the achievable distortion-cost region simplifies to the expression given in the Proposition. 

For the converse, fixing a $(n, C)$ code achieving distortion $D$, we have
\begin{align*}
D &= \frac{1}{n} \sum_{i=1}^n \E d(S_{1i}, \Sh_{1i}(Y^n)).
\end{align*}
Consider now the term $\E d(S_{1i}, \Sh_{1i}(Y^n))$. We have{\allowdisplaybreaks
\begin{align*}
\E d(S_{1i}, \Sh_{1i}(Y^n)) &= \sum p(s_{1i}, y^{n\backslash i}, y_i) d(s_{1i}, \sh_{1i}(y^n)) \\
& = \sum \left(p(s_{1i}, y^{n\backslash i}, y_i, x\oplus s_{2i} = 0) d(s_{1i}, \sh_{1i}(y^n))\right.\\
& \left.\qquad \quad + p(s_{1i}, y^{n\backslash i}, y_i, x\oplus s_{2i} = 1) d(s_{1i}, \sh_{1i}(y^n))\right) \\
& \stackrel{(a)}{=} \sum \left(p(s_{1i}, y^{n\backslash i}, y_i, x\oplus s_{2i} = 0) d(s_{1i}, \sh_{1i}(y^n))\right.\\
& \left.\qquad \quad + p(s_{1i}, y^{n\backslash i}, x\oplus s_{2i} = 1) d(s_{1i}, \sh_{1i}(y^{n\backslash i}, y_i = e))\right) \\
& = \sum p(s_{1i}, x\oplus s_{2i} = 0)p( y^{n\backslash i}, y_i|s_{1i}, x\oplus s_{2i} = 0) d(s_{1i}, \sh_{1i}(y^n))\\
& \qquad + \sum p(s_{1i}, y^{n\backslash i}, x\oplus s_{2i} = 1) d(s_{1i}, \sh_{1i}(y^{n\backslash i}, y_i = e)),
\end{align*}}
where $(a)$ follows from the fact that when $x \oplus s_{2i} = 1$, $y_i = e$. Next, focusing on the first term in the sum, we note that $\P(Y_i = S_{1i}|X\oplus S_2 = 0, S_{1i}) = 1$. Hence, using $1_{\{.\}}$ to denote the indicator function, the first term simplifies to the following
\begin{align*}
& \sum p(s_{1i}, x\oplus s_{2i} = 0)p( y^{n\backslash i}, y_i|s_{1i}, x\oplus s_{2i} = 0) d(s_{1i}, \sh_{1i}(y^n)) \\
& = \sum p(s_{1i}, x\oplus s_{2i} = 0)p( y^{n\backslash i}|s_{1i}, x\oplus s_{2i} = 0, y_i)1_{y_i = s_{1i}} d(s_{1i}, \sh_{1i}(y^n)) \\
& \ge \sum p(s_{1i}, x\oplus s_{2i} = 0)p( y^{n\backslash i}|s_{1i}, x\oplus s_{2i} = 0, y_i)\left(\min_{x \in \Sh_1} d(s_{1i}, x)\right) \\
& = \sum p(s_{1i})p( x\oplus s_{2i} = 0)\left(\min_{x \in \Sh_1} d(s_{1i}, x)\right) \\
& =\P( X_i\oplus S_{2i} = 0)\E \left(\min_{x \in \Sh_1} d(s_{1i}, x)\right).
\end{align*}

Hence, $\E d(S_{1i}, \Sh_{1i}(Y^n))$ is lower bounded by{\allowdisplaybreaks
\begin{align*}
\E d(S_{1i}, \Sh_{1i}(Y^n)) & \ge \P( X_i\oplus S_{2i} = 0)\E \left(\min_{x \in \Sh_1} d(s_{1i}, x)\right)\\
& \quad + \sum p(s_{1i}, y^{n\backslash i}, x\oplus s_{2i} = 1) d(s_{1i}, \sh_{1i}(y^{n\backslash i}, y_i = e)), \\
& = \P( X_i\oplus S_{2i} = 0)\E \left(\min_{x \in \Sh_1} d(s_{1i}, x)\right)\\
& \quad + \sum p(s_{1i})p( y^{n\backslash i}, x\oplus s_{2i} = 1) d(s_{1i}, \sh_{1i}(y^{n\backslash i}, y_i = e)) \\
& \ge \P( X_i\oplus S_{2i} = 0)\E \left(\min_{x \in \Sh_1} d(s_{1i}, x)\right)\\
& \quad + \sum p( y^{n\backslash i}, x\oplus s_{2i} = 1) \left(\min_{x \in \Sh_1}\sum p(s_{1i}) d(s_{1i}, x)\right) \\
& = \P( X_i\oplus S_{2i} = 0)\E \left(\min_{x \in \Sh_1} d(s_{1i}, x)\right)\\
& \quad + \sum \P(X_i\oplus S_{2i} = 1)\left( \min_{x \in \Sh_1}\E d(S_{1i}, x)\right).
\end{align*}}
We note now that if $y_i = 0$ or $1$, then we can achieve the minimum possible distortion $\min_{\sh_1} d(s_1, \sh_1)$ using only knowledge of $y_i$, since $s_{1i}$ is known in this case.

We therefore obtain
\begin{align*}
D &= \frac{1}{n} \sum_{i=1}^n \E d(S_{1i}, \Sh_{1i}(Y^n)) \\
& \ge \frac{1}{n} \sum_{i=1}^n  \left(\P( X_i\oplus S_{2i} = 0)\E \left(\min_{x \in \Sh_1} d(s_{1i}, x)\right)\right.\\
& \left.\qquad \qquad  + \sum \P(X_i\oplus S_{2i} = 1)\left( \min_{x \in \Sh_1}\E d(S_{1i}, x)\right)\right).
\end{align*}
Defining $Q \sim Unif[1:n]$ independent of other random variables then give us
\begin{align*}
D &\ge \P(X \oplus S_2 = 1)(\min_{\sh_1} \E d(S_1, \sh_1)) \\
& \qquad + \P(X \oplus S_2 = 0)\E (\min_{\sh_1} d(S_1, \sh_1)).
\end{align*}
For the cost constraint, we have
\begin{align*}
\E( \frac{1}{n} \sum_{i=1}^n X_i) & = \E X_Q \\
& = \E X,
\end{align*}
which completes the proof. 
\end{proof}
\section{Gaussian Estimation with helper} \label{sect:5}
In this section, we extend our setup to the Gaussian case, where $S_1 \sim N(0,1)$, $S_2 \sim N(0, P_2)$, $Y = X+S_1 + S_2$, $d(S_1, \Sh_1) = (S_1 - \Sh_1)^2$ and the cost constraint is $\E X^2 \le P$. As we mentioned in the Introduction, the problem in the Gaussian case is equivalent to the problem of Assisted Interference Suppression considered in \cite{Grover--Sahai2011}. We present a new lower bound for this problem that can improve on that derived in \cite{Grover--Sahai2011} and \cite{Grover--Wagner--Sahai2011}. The lower bound derived in \cite{Grover--Wagner--Sahai2011} includes the lower bound derived in \cite{Grover--Sahai2011} as a special case and can be strictly better, but for clarity of presentation, we will first compare our lower bound to that in \cite{Grover--Sahai2011} in subsection \ref{subsect:comgs}, and then compare our bound with the lower bound derived in \cite{Grover--Wagner--Sahai2011} in subsection \ref{subsect:comgws}. We begin with an achievability argument based on Theorem \ref{thm:2}. 

\subsection{Achievable distortion-cost region}
We specialize Theorem~\ref{thm:2} to the Gaussian case by choosing the auxiliary random variables as Gaussian random variables. The achievability scheme presented here is essentially the same as the scheme presented in \cite{Grover--Sahai2011}, but we derive it via different means. 
\begin{theorem} \label{thm:5}
An achievable distortion for the problem of Gaussian estimation with a helper is given by
\begin{align*}
D(P)_{\rm min} \le \inf \; 1 - \frac{\E U^2}{\E Y^2 \E U^2 - (\E UY)^2},
\end{align*}
where
\begin{align*}
\sqrt{P'} &= \frac{-2\alpha \beta \sqrt{P} + \sqrt{4\alpha^2 \beta^2 P + 4(1-\alpha^2)P}}{2}, \\
\E U^2 &= P' + 2\gamma \beta \sqrt{P' P_2} + \gamma^2 P_2, \\
\E(UY) & = \alpha \beta \sqrt{P P'} + P' + \alpha \gamma \sqrt{P P_2} + \gamma \beta \sqrt{P' P_2}+ \beta\sqrt{P' P_2} + \gamma P_2,\\
\E Y^2 & = P + 1 + P_2  + 2 \alpha \sqrt{P P_2} + 2\beta \sqrt{P' P_2}.
\end{align*}
and the infinum is taken over $-1 \le \alpha \le 1$, $-1 \le \beta \le 1$ and $\gamma \in \Rc$ satisfying the constraint
\begin{align*}
(1-\beta^2)P' > \E U^2 - \frac{(\E (UY))^2}{\E Y^2}.
\end{align*}
\end{theorem}

We defer the proof of Theorem \ref{thm:5} to Appendix~\ref{appen_g1}.

Similar to the binary setup, we can derive a nontrivial condition between $P$ and the power of the source $S_1$ (normalized to 1), such that zero expected distortion can be achieved.

\begin{proposition} \label{prop5}
For the problem of Gaussian estimation with a helper, $D(C)_{\rm min} = 0$ if
\begin{align*}
P > 1- \frac{1}{P+P_2 +1}.
\end{align*}
\end{proposition}
\begin{proof}
Proof of this Proposition follows from a choice of $\alpha$ and $\beta$ in Theorem~\ref{thm:5}. However, we give a slightly different proof that gives more intuition to this condition and also has parallels with the problem of dirty paper coding~\cite{Costa1983} (see also \cite[Chapter 7]{El-Gamal--Kim2010}).

Starting from Theorem~\ref{thm:2}, we let $U = X+S_2$, where $X \sim N(0, P)$ independent of $S_2$. Note that the cost constraint is satisfied from this choice of $U$. If the decoder can decode $U$, then the distortion incurred is zero, since $S_1 = Y-U$. It therefore remains to satisfy the decoding condition, which is
\begin{align*}
I(U;Y) > I(U;S_2).
\end{align*}
Since all the random variables are Gaussian, this decoding condition reduces to $h(U|S_2) > h(U|Y)$. 
\begin{align*}
h(U|S_2) &= h(X|S_2) \\
& = \frac{1}{2}\log 2 \pi e P.
\end{align*}
On the other hand,
\begin{align*}
h(U|Y) &= h(-S_1|Y) \\
& \stackrel{(a)}{=} h(-S_1 -\E(-S_1|Y)) \\
& = \frac{1}{2}\log 2\pi e \left( 1- \frac{1}{P_1+P_2 +1}\right),
\end{align*}
where $(a)$follows from the fact that for Gaussian random variables, the difference between $S_1$ and its Minimum Mean Square Error Estimator is independent of the observation, $Y$. 

We therefore derive the condition
\begin{align*}
P > 1 - \frac{1}{P+P_2 +1}.
\end{align*}
\end{proof}

Note that, similar to the binary case, the expected distortion can be made to be zero even if $P_2$ is much larger than $P$. 
\subsection{Lower bounds}
We now turn to lower bounds for the problem of Gaussian Estimation with helper. We first state the following lower bound given in \cite{Grover--Sahai2011} and its improved version given in \cite{Grover--Wagner--Sahai2011}.
\begin{theorem} \label{thm:6}
~\cite{Grover--Sahai2011} A lower bound for the problem of Gaussian estimation with helper is given by
\begin{align*}
D(P)_{\rm min} \ge \left(\left[\sqrt{\frac{P_2}{P_22 + 2\sqrt{PP_2} + P +1}}-\sqrt{P}\right]^+\right)^2,
\end{align*} 
where $[.]^+$ denotes the positive part. 
\end{theorem}

As shown in \cite{Grover--Wagner--Sahai2011}, the lower bound given in Theorem~\ref{thm:6} can be improved to the following.

\begin{theorem} \label{thm:GWS}
~\cite{Grover--Wagner--Sahai2011} A lower bound for the problem of Gaussian estimation with helper is given by
\begin{align*}
D(P)_{\rm min} \ge \inf_{\sigma_{XS_2}} \sup_{\gamma >0} \frac{1}{\gamma^2} \left(\left[ \sqrt{\frac{P_2}{1+ P_2 + P + 2\sigma_{XS_2}}} - \sqrt{(1-\gamma)^2P_2 + \gamma^2 P -2\gamma(1-\gamma)\sigma_{XS_2}}\right]^+\right)^2,
\end{align*} 
where $[.]^+$ denotes the positive part and $\sigma_{\rm XS_2} \in [-\sqrt{P_2}\sqrt{P}, \sqrt{P_2}\sqrt{P}]$. 
\end{theorem}

From the lower bound in Theorem~\ref{thm:6} and Proposition~\ref{prop5}, we can show that as the power of the interfering signal goes to infinity, $P_2 \to \infty$, zero expected distortion is achievable if and only if $P \ge1$.

\begin{proposition} \label{prop_hi}
$\lim_{P_2\rightarrow\infty}D(P)_{\rm min} = 0$ if and only if $P\geq 1$. 
\end{proposition} 
\begin{proof}
From Proposition~\ref{prop5}, the sufficient condition for zero expected distortion reduces to $P > 1$ as $P_2 \to \infty$. From Theorem~\ref{thm:6}, we can show that this is also necessary. From Theorem~\ref{thm:6}, $\lim_{P_2 \to \infty} D(P)_{\rm min} \ge \left(\left[1-\sqrt{P}\right]^+\right)^2$, which is zero if and only if $P \ge 1$.
\end{proof}

We now turn to our lower bound. For clarity, we first present a proof of a special case of our lower bound before turning to the more general expression. 

\begin{proposition} \label{prop6}
A lower bound for the problem of Gaussian estimation with a helper is given by
\begin{align*}
D(P)_{\rm min} &\ge  \frac{1}{(\gamma -1)}\left[\ln \left(\frac{1+\gamma P_2}{1+  P_2}\right) + \frac{2 \sqrt{P}}{\sqrt{P_2} (1+ \gamma P_2)} - \frac{2 \sqrt{P}}{\sqrt{P_2} (1+  P_2)} -  \gamma P\right],
\end{align*} 
for \textit{any} $\gamma \ge 1$.
\end{proposition}
It should be noted that while finding the optimal value of $\gamma$ that \textit{maximizes} this lower bound is a hard optimization problem, \textit{any} $\gamma \ge 1$ constitutes a lower bound for $D(P)_{\rm min}$. Hence, Proposition~\ref{prop6} in fact gives a \textit{family} of lower bounds. 

\begin{proof}
This proof hinges on an application of a relationship between mismatched estimation and relative entropy given in \cite[Equality (14)]{Verdu2010}. The main idea behind the proof lies in considering a decoder that performs the estimation (and reconstruction) using a wrong (or mismatched) distribution for $\P_{S_1^n|Y^n}$. In particular, we will consider a \textit{mismatched} decoder that attempts to estimate $S_1^n$ assuming that $X^n \equiv 0$. That is, the decoder assumes that the encoder does not do anything to help the decoder. The estimation error incurred by the mismatched decoder, $MSE_Q$, is clearly larger than that incurred by an optimum decoder that uses the correct (true) distribution, $D(P)_{\rm min}$. We then rely on results in \cite{Verdu2010} to lower bound the difference between $D(P)_{\rm min}$ and $MSE_Q$, thereby giving us a lower bound on $D(P)_{\rm min}$.  

To derive our bound, we first consider a more general source $S_1 \sim N(0, 1/\gamma)$ and let $S_2 \sim N(0, P_2)$ as before. The value of $\gamma$ that we are concerned about is $\gamma =1$, which will appear later in the proof.  

Define $MSE_Q (\gamma)$ as
\begin{align*}
MSE_Q (\gamma) : = \E \bigg|\bigg|S_1^n - \frac{\frac{1}{\gamma}}{\frac{1}{\gamma} + P_2}(X^n + S_1^n + S_2^n)\bigg|\bigg|^2. 
\end{align*}
Let $\alpha = \frac{\frac{1}{\gamma}}{\frac{1}{\gamma} + P_2}$ and note that $\Sh_{1} = \alpha Y$ is the Minimum Mean Square Error (MMSE) estimate of $S_1$ that the decoder would employ if it assumes that $X^n \equiv 0$. We first give a lower bound for $MSE_Q(\gamma)$. Note that under the true distribution, $\E||X^n||^2 \le nP$.
\begin{align*}
\E||S_1^n - \alpha(X^n + S_1^n + S_2^n)||^2 &= \E||S_1^n - \alpha(S_1^n + S_2^n) ||^2 - 2 \alpha \E <S_1^n - \alpha(S_1^n+S_2^n), X^n> + \alpha^2 \E||X^n ||^2 \\
& = n \alpha P_2 + 2 \alpha^2 \E <S_2^n, X^n> + \alpha^2 \E ||X^n||^2 \\
& \stackrel{(a)}{\ge} n \alpha P_2 - 2 \alpha^2 \sqrt {\E ||S^n_2||^2 \E ||X^n||^2} + \alpha^2 \E ||X^n||^2\\
& \ge n \alpha P_2 - 2 \alpha^2 \sqrt {n^2 P_2 P} \\
& = n \alpha P_2 - 2 n\alpha^2 \sqrt { P_2 P},
\end{align*}
where $(a)$ follows by Cauchy-Schwartz inequality.

Now, let $\St^n = S_2^n + X^n$ and let $P_{\St^n}$ denote the distribution of $S_2^n + X^n$ under the optimum encoding scheme. Let $Q_{\St^n}$ denote the corresponding distribution under the encoding scheme of $X^n \equiv 0$. Note now that
\begin{align}
MSE_Q (\gamma) & = \E ||S_1^n - \alpha(S_1^n + \St^n)||^2 \nonumber\\
& \stackrel{(a)}{=} \E || Y^n - \St^n - \E_Q(S_1^n|Y^n)||^2 \nonumber\\
& \stackrel{(b)}{=} \E || Y^n - \St^n - (Y^n- \E_Q(\St^n|Y^n)||^2 \nonumber\\
& = \E ||\St^n - \E_Q(\St^n|Y^n)||^2 \nonumber\\
& := MSE_{Q, \St^n} (\gamma). \label{eqm1}
\end{align}
$(a)$ follows from the fact that $\alpha(S_1^n + \St^n)$ is the optimum MMSE estimator for $S_1^n$ under $Q$; that is, under the assumption of $X^n \equiv 0$. $(b)$ follows from $S_1^n = Y^n - \St^n$.

Next, note that this analysis also holds when the decoder knows that $\St^n$ is distributed according to $P_{\St^n}$. That is, we have 
\begin{align}
MMSE(\gamma)&:= \E||S_1^n - \E_P(S_1^n|Y^n)||^2 \nonumber\\
& =  \E ||\St^n - \E_P(\St^n|Y^n)||^2 \nonumber\\
& := MMSE_{P, \St^n}(\gamma). \label{eqm2}
\end{align}
Note that $nD(P)_{\rm min} = MMSE(1)$.

We now relate $MSE_Q(\gamma)$ to the optimum $MMSE$ of $S_1^n$ given that an optimum estimator and coding scheme were used. From \eqref{eqm1} and \eqref{eqm2}, we see that it suffices to consider $MSE_{Q, \St^n} (\gamma)$ and $MMSE_{P, \St^n}(\gamma)$. Using the relation between mismatched estimation and relative entropy given in \cite[Equality 14]{Verdu2010}, we have
\begin{align}
D(P_{Y^n}^{(\gamma_0)}||Q_{Y^n}^{(\gamma_0)}) = \frac{1}{2}\int_{0}^{\gamma_0} MSE_{Q, \St^n} (\gamma) - MMSE_{P, \St^n} (\gamma) d\gamma \label{verdu1}
\end{align}
Here, $P_{Y^n}$ represents the distribution of $Y^n$ induced by $P_{\St^n}$. Similarly, $Q_{Y^n}$ represents the distribution of $Y^n$ induced by $Q_{\St^n}$.

We first give a bound on $D(P_{Y^n}^{(\gamma)}||Q_{Y^n}^{(\gamma)})$. 
\begin{align*}
D(P_{Y^n}^{(\gamma)}||Q_{Y^n}^{(\gamma)}) &\le D(P_{Y^n, S_2^n}^{(\gamma)}||Q_{Y^n, S_2^n}^{(\gamma)}) \\
& = \E_{S_2^n} D(P_{Y^n| S_2^n}^{(\gamma)}||Q_{Y^n| S_2^n}^{(\gamma)}).
\end{align*}
Note that since $X^n$ is a function of $S^n_2$, we have the following.
\begin{align*}
&\mbox{Under } P_{Y^n| S_2^n}^{(\gamma)}: Y^n|S_2^n \sim N(S_2^n + X^n, \frac{1}{\gamma}I_{n\times n}), \\
&\mbox{Under } Q_{Y^n| S_2^n}^{(\gamma)}: Y^n|S_2^n \sim N(S_2^n, \frac{1}{\gamma}I_{n\times n}).
\end{align*}

Hence, $D(P_{Y^n| S_2^n}^{(\gamma)}||Q_{Y^n| S_2^n}^{(\gamma)})$ is given by the divergence between two multivariate Gaussian random variables with the same covariance matrix. In our case, the divergence is given by
\begin{align*}
D(P_{Y^n| S_2^n}^{(\gamma)}||Q_{Y^n| S_2^n}^{(\gamma)}) = \frac{1}{2} \gamma ||X^n||^2.
\end{align*}

Hence, 
\begin{align*}
D(P_{Y^n}^{(\gamma)}||Q_{Y^n}^{(\gamma)}) & \le \E_{S_2^n} D(P_{Y^n| S_2^n}^{(\gamma)}||Q_{Y^n| S_2^n}^{(\gamma)}) \\
& = \E_{S_2^n}( \frac{\gamma}{2}  ||X^n||^2)\\
& \le \frac{n\gamma P}{2} 
\end{align*}
From \eqref{verdu1}, we have
\begin{align*}
\int_{\gamma_0}^{\gamma_1} MSE_{Q, \St^n} (\gamma) - MMSE_{P, \St^n} (\gamma) d\gamma = 2D(P_{Y^n}^{(\gamma_1)}||Q_{Y^n}^{(\gamma_1)}) - 2D(P_{Y^n}^{(\gamma_0)}||Q_{Y^n}^{(\gamma_0)})
\end{align*}
for $\gamma_1 \ge \gamma_0$.
Hence,
\begin{align}
\int_{\gamma_0}^{\gamma_1} MMSE_{P, \St^n} (\gamma) d\gamma &\ge \int_{\gamma_0}^{\gamma_1} MSE_{Q, \St^n} (\gamma) d\gamma - 2D(P_{Y^n}^{(\gamma_1)}||Q_{Y^n}^{(\gamma_1)})  \nonumber \\
&\ge  \int_{\gamma_0}^{\gamma_1} MSE_{Q, \St^n} (\gamma) d\gamma - n\gamma_1P \label{eqn:1}
\end{align}
Since $MMSE_{P, \St^n} (\gamma) $ is a non-increasing function in $\gamma$, we have $\int_{\gamma_0}^{\gamma_1} MMSE_{P, \St^n} (\gamma) d\gamma \le (\gamma_1 - \gamma_0)MMSE_{P, \St^n} (\gamma_0) = (\gamma_1 - \gamma_0)MMSE (\gamma_0)$. Next, we note that $\alpha = 1/(1+\gamma P_2)$, so we can write
\begin{align*}
MSE_{Q, \St^n} (\gamma) &= MSE_{Q} (\gamma) \\
& \ge \frac{n P_2}{1+ \gamma P_2} - \frac{2 n \sqrt { P_2 P}}{(1+\gamma P_2)^2}. 
\end{align*}

From \eqref{eqn:1} and the arguments above, we have
\begin{align*}
 (\gamma_1 - \gamma_0)MMSE (\gamma_0) &\ge  \int_{\gamma_0}^{\gamma_1} \frac{n P_2}{1+ \gamma P_2} - \frac{2 n \sqrt { P_2 P}}{(1+\gamma P_2)^2} d\gamma - n\gamma_1P \\
 & = n \ln (\frac{1+\gamma_1 P_2}{1+ \gamma_0 P_2}) + \frac{2n \sqrt{P}}{\sqrt{P_2} (1+ \gamma_1 P_2)} - \frac{2n \sqrt{P}}{\sqrt{P_2} (1+ \gamma_0 P_2)} -  n\gamma_1P. 
\end{align*}
Finally, using the relationship that $D(P)_{\rm min} = MMSE(1)/n$, $\gamma_0=1$ and the above completes the proof of the lower bound. 
\end{proof}
In Proposition~\ref{prop6}, we related the minimum mean square error that a decoder incurs when it uses the true distribution to the mean square error incurred by a decoder if it uses the possibly erroneous distribution of $X^n \equiv 0$. Clearly, we do not need to choose $X^n \equiv 0$ as the erroneous distribution, but we can also choose other distributions. This is the main idea behind our generalization of Proposition~\ref{prop6}, which we state in Theorem~\ref{thm:7}. 
\begin{theorem} \label{thm:7}
A lower bound for the problem of Gaussian estimation with helper is given by
\begin{align*}
(\gamma - 1)D(P)_{\rm min} &\ge \log(\frac{1+ \gamma P_I}{1+  P_I}) + \frac{1}{(1+ \gamma P_I)} - \frac{1}{(1+  P_I)} \\
& \quad -\frac{P_2}{P_I(1+ \gamma P_I)}+\frac{P_2}{P_I(1+  P_I)} - \frac{c^2\gamma}{1+\gamma rP}P_2 \\
& \quad -\frac{1}{1+ \gamma_1 rP} +1 + \log(\frac{1}{1+ \gamma_1 rP}) \\
& \quad + a x^{*2} - b x^*,
\end{align*}
where $a = \frac{1}{P_I (1+  P_I)} - \frac{1}{P_I (1+ \gamma P_I)} - \frac{\gamma}{1+\gamma rP}$, $b = |2(\frac{1}{P_I (1+  P_I)} - \frac{1}{P_I (1+ \gamma P_I)} +  \frac{c\gamma}{1+\gamma rP})|\sqrt{P_2}$ and
\begin{align*}
x^* = \left\{\begin{array}{cl} \sqrt{P} & \mbox{if } a \le 0 \\
b/2a & \mbox{if } a >0 \mbox{ and } b/2a < \sqrt{P} \\
\sqrt{P} & \mbox{otherwise}   
\end{array} \right. , 
\end{align*}
for \textit{any} $\gamma \ge 1$, real number $c$ and $r \ge 0$.
\end{theorem}
As with Proposition~\ref{prop6}, Theorem~\ref{thm:7} gives a family of bounds. \textit{Any} $\gamma > 1$, real number $c$ and $r \ge 0$ yields a bound on the achievable distortion. Theorem~\ref{thm:7} is proved in Appendix~\ref{appen_glb}.

\subsection{Comparison of bounds I} \label{subsect:comgs}
We now show some plots comparing the various bounds we derived with the lower bound proposed in \cite{Grover--Sahai2011} (Theorem \ref{thm:6}). For the purpose of comparisons, we set $P_2$ at a fixed level and vary the power of the encoder. We then compute the lower bounds on distortion given in Theorem~\ref{thm:5}, Proposition~\ref{prop6}, Theorem~\ref{thm:7} as well as the achievable distortion given in Theorem~\ref{thm:5}.

The plots for $P_2 = 0.1$, $P_2 = 1$ and $P_2 = 10$ are shown in Figures~\ref{fig:v1},~\ref{fig:v2} and~\ref{fig:v3} respectively. As we can see from the plots, the generalized lower bound in Theorem~\ref{thm:7} can significantly improve on the lower bound of Theorem~\ref{thm:5} for several different levels of $P_2$. 

\begin{figure}[!h]
\psfrag{Achievable Gaussian}[l]{Theorem~\ref{thm:5}}
\psfrag{New LB 2}[l]{Theorem~\ref{thm:7}}
\psfrag{New LB}[l]{Proposition~\ref{prop6}}
\psfrag{Grover-Sahai LB}[l]{Theorem~\ref{thm:6}}

\begin{center}
\scalebox{0.75}{\includegraphics{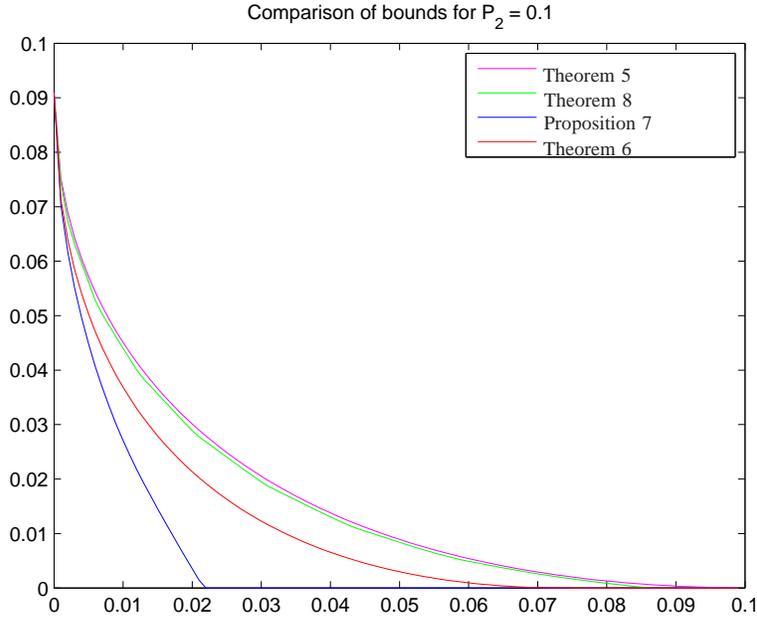}}
\caption{Comparison of bounds for $P_2 = 0.1$. Y-axis represents distortion level and X-axis represents the power constraint.} \label{fig:v1}
\end{center}
\end{figure}  

\begin{figure}[!h]
\psfrag{Achievable Gaussian}[l]{Theorem~\ref{thm:5}}
\psfrag{New LB 2}[l]{Theorem~\ref{thm:7}}
\psfrag{New LB}[l]{Proposition~\ref{prop6}}
\psfrag{Grover-Sahai LB}[l]{Theorem~\ref{thm:6}}

\begin{center}
\scalebox{0.75}{\includegraphics{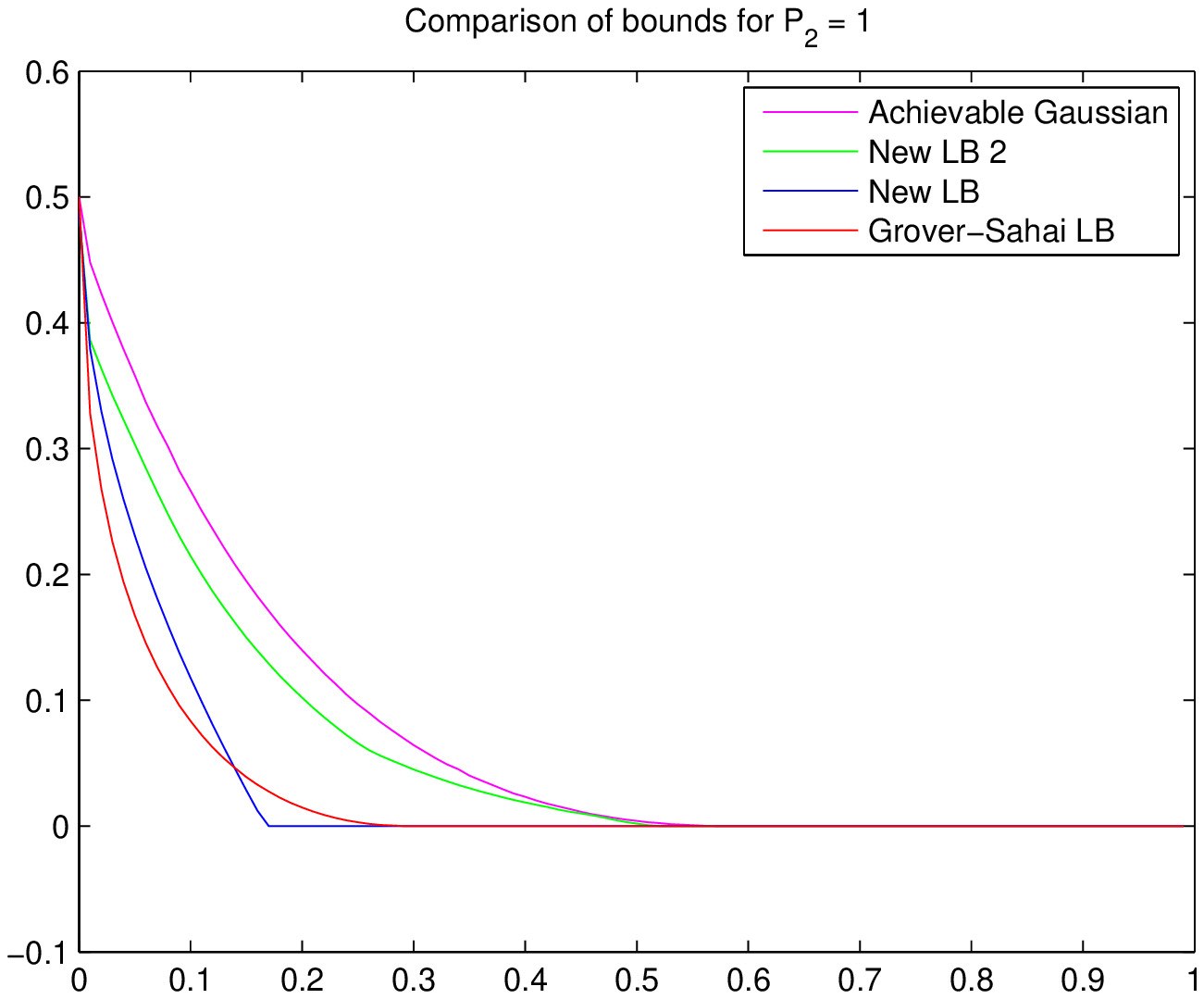}}
\caption{Comparison of bounds for $P_2 = 1$. Y-axis represents distortion level and X-axis represents the power constraint.} \label{fig:v2}
\end{center}
\end{figure}  

\begin{figure}[!h]
\psfrag{Achievable Gaussian}[l]{Theorem~\ref{thm:5}}
\psfrag{New LB 2}[l]{Theorem~\ref{thm:7}}
\psfrag{New LB}[l]{Proposition~\ref{prop6}}
\psfrag{Grover-Sahai LB}[l]{Theorem~\ref{thm:6}}

\begin{center}
\scalebox{0.75}{\includegraphics{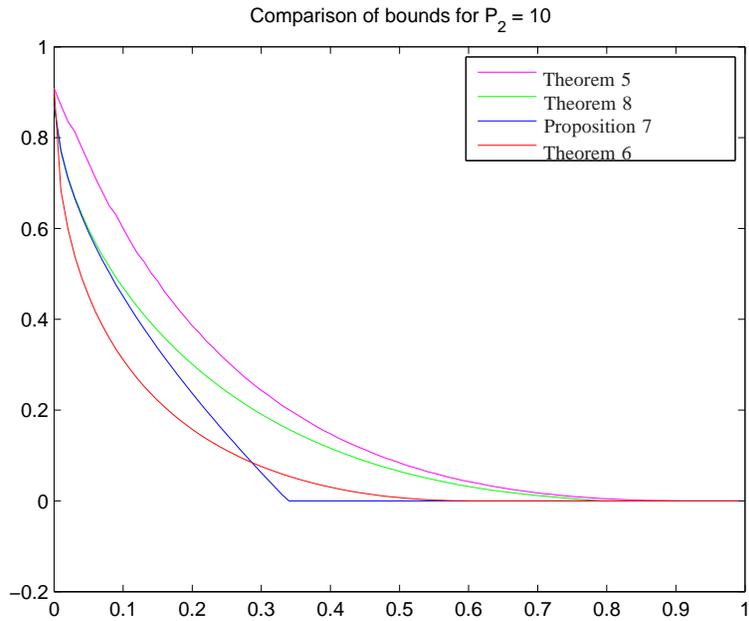}}
\caption{Comparison of bounds for $P_2 = 10$. Y-axis represents distortion level and X-axis represents the power constraint.}\label{fig:v3}
\end{center}
\end{figure}  
\subsection{Comparison of bounds II} \label{subsect:comgws}

In this subsection, we compare our lower bound given in Theorem~\ref{thm:7} to the lower bound given in \cite{Grover--Wagner--Sahai2011} (Theorem \ref{thm:GWS}). For ease of numerical computation, we compare our lower bound to the following \textit{upper bound} on Theorem~\ref{thm:GWS}.
\begin{align}
D(P)_{\rm min} \ge \min_{\sigma_{XS_2} \in \mathcal{A}} \sup_{\gamma >0} \frac{1}{\gamma^2} \left(\left[ \sqrt{\frac{P_2}{1+ P_2 + P + 2\sigma_{XS_2}}} - \sqrt{(1-\gamma)^2P_2 + \gamma^2 P -2\gamma(1-\gamma)\sigma_{XS_2}}\right]^+\right)^2, \label{ineq:GWS}
\end{align} 
where $[.]^+$ denotes the positive part and $\mathcal{A}$ is a discretization of the interval $[-\sqrt{P_2}\sqrt{P}, \sqrt{P_2}\sqrt{P}]$. The plots showing comparisons of the lower bound proposed in Theorem~\ref{thm:7} and the lower bound given in inequality~\eqref{ineq:GWS} for $P_2=1,10,100$ are given in Figures~\ref{fig:v4},~\ref{fig:v5} and~\ref{fig:v6} respectively.

As can be seen from the plots, the two bounds now cross each other. While the lower bound given~\cite{Grover--Wagner--Sahai2011} can be better than that given in Theorem~\ref{thm:7} in some regimes, we can also see that Theorem~\ref{thm:7} can be strictly better than Theorem \ref{thm:GWS} in other regimes, particularly when $P_2$ is large and the power budget $P$ of the encoder is small.  

\begin{figure}[!h]
\psfrag{New LB}[l]{Theorem~\ref{thm:7}}
\psfrag{Grover-Sahai LB}[l]{Inequality~\eqref{ineq:GWS}}

\begin{center}
\scalebox{0.75}{\includegraphics{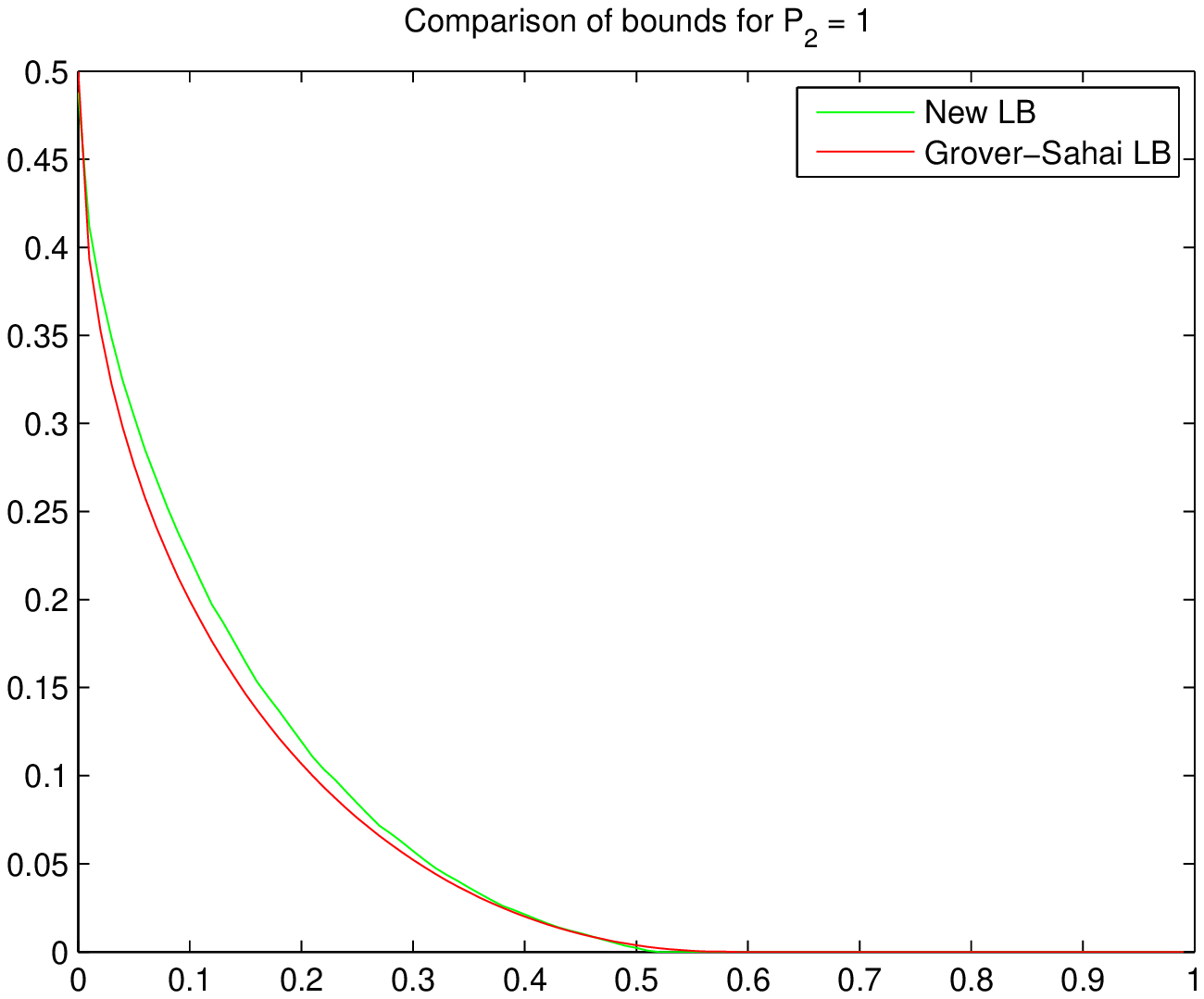}}
\caption{Comparison of bounds for $P_2 = 1$. Y-axis represents distortion level and X-axis represents the power constraint.} \label{fig:v4}
\end{center}
\end{figure}  

\begin{figure}[!h]
\psfrag{New LB}[l]{Theorem~\ref{thm:7}}
\psfrag{Grover-Sahai LB}[l]{Inequality~\eqref{ineq:GWS}}

\begin{center}
\scalebox{0.75}{\includegraphics{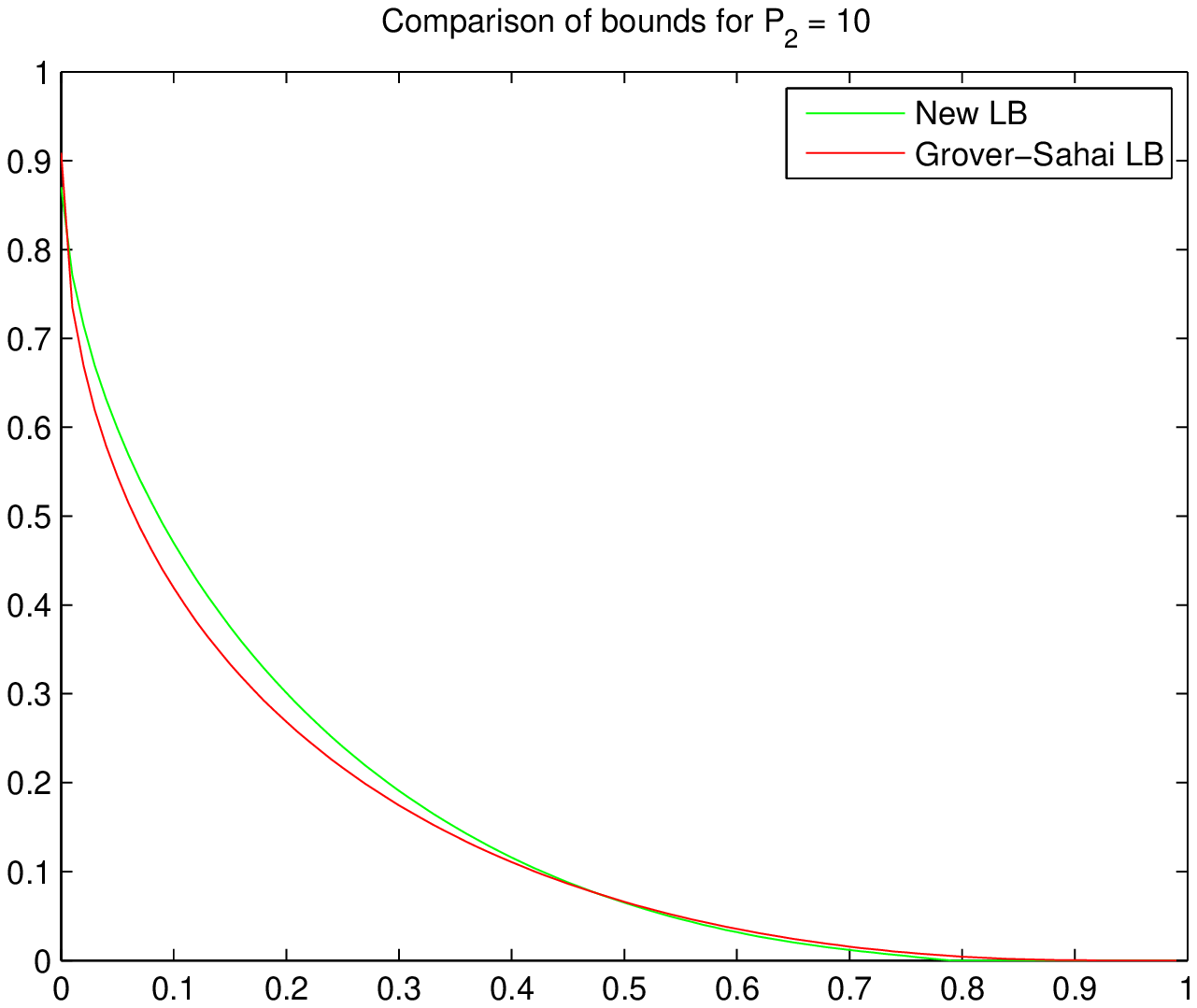}}
\caption{Comparison of bounds for $P_2 = 10$. Y-axis represents distortion level and X-axis represents the power constraint.} \label{fig:v5}
\end{center}
\end{figure}  

\begin{figure}[!h]
\psfrag{New LB}[l]{Theorem~\ref{thm:7}}
\psfrag{Grover-Sahai LB}[l]{Inequality~\eqref{ineq:GWS}}

\begin{center}
\scalebox{0.75}{\includegraphics{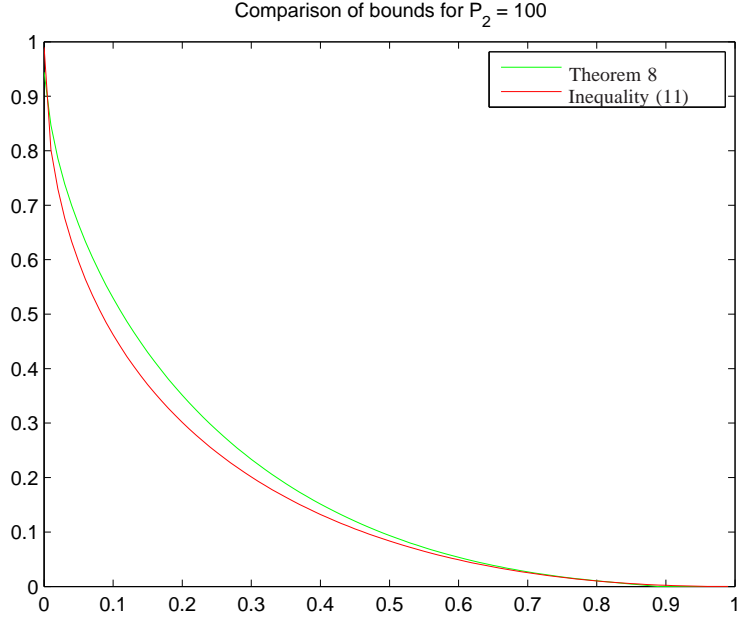}}
\caption{Comparison of bounds for $P_2 = 100$. Y-axis represents distortion level and X-axis represents the power constraint.} \label{fig:v6}
\end{center}
\end{figure}

\section{When $S_1$ is also available at the encoder} \label{sect:6}
In this section, we turn our attention to the problem of reconstructing $S_1$ when both $S_1$ and $S_2$ are available at the encoder, as defined in Section~\ref{sect:svprob}. As with previous sections, the focus of this section is on lower bounds for this setup, but we also use lower and upper bounds to derive constant multiplicative gap results between the achievable distortions and lower bounds. As we mentioned in the Introduction, our setting is a special case of the setting considered in~\cite{Huang--Narayanan2011}. We first review some known results found in that paper specialized to our setting, and then present our results, which include a generalization of the lower bound~\cite{Huang--Narayanan2011} that can be strictly larger.  

\subsection{Upper and lower bounds}
We first present an achievability scheme for this setting.

\begin{theorem} \label{thm:8}
(See also~\cite{Huang--Narayanan2011}) An acheivable distortion-cost region for the problem of estimation with a helper who has non-causal access to both the interference and the signal is given by
\begin{align*}
D(P)_{\rm min} \le \frac{P_1}{\left(1+\frac{\left(\alpha\sqrt{\frac{P}{P_1}}+1\right)^2P_1}{ \left(\beta\sqrt{\frac{P}{P_2}}+1\right)^2P_2+P(1-\alpha^2-\beta^2)+N}\right)\left(1+\frac{P(1-\alpha^2-\beta^2)}{N}\right)},
\end{align*}
where we minimize over $-1\le \alpha \le 1$, $-1 \le \beta \le 1$ and $0 \le \alpha^2 + \beta^2 \le 1$.\footnote{In \cite{Huang--Narayanan2011}, the authors minimize only over $0\le \alpha \le 1$, $0 \le \beta \le 1$ and $0 \le \alpha^2 + \beta^2 \le 1$, but it is easy to see that their proof carries over to the range stated in this Theorem.}
\end{theorem}

As the achievability scheme is largely the same as that in \cite{Huang--Narayanan2011}, we only give a sketch in Appendix~\ref{appen_gsv}.

We now turn to lower bounds on the distortion-cost region. We first present without proof two lower bounds in the following two propositions. For their proofs cf.  \cite{Huang--Narayanan2011} or the proof of Theorem~\ref{thm:9} below. 

\begin{proposition} \label{prop7}
A lower bound for  the problem of estimation with a helper who knows both the interference and the signal noncausally is given by
\begin{align*}
D(P)_{\rm min} \ge \frac{P_1}{\left(1+ \frac{P_1}{P_2}\right)\left(1+ \frac{P}{N}\right)}.
\end{align*}
\end{proposition} 
\begin{remark}
When $P_2 \to \infty$, we see that $D(P)_{\rm min}\ge\frac{P_1}{1+\frac{P}{N}}$. This bound is achievable by noting that for $\alpha=\beta=0$, we have $D=\frac{P_1}{1+\frac{P}{N}}$ in Theorem~\ref{thm:8}. Thus a separation scheme is optimal when $P_2\rightarrow\infty$. 
\end{remark}

\begin{proposition} \label{prop8}
A lower bound for  the problem of estimation with a helper who knows both the interference and the signal noncausally is given by
\begin{align*}
D(P)_{\rm min} \ge \frac{P_1}{\left(1+ \frac{(\sqrt{P} + \sqrt{P_1})^2}{N}\right)}.
\end{align*}
\end{proposition}
 \begin{remark}
As $P_2 \to 0$, our setting reduces to that of state amplification~\cite{Kim--Sutivong--Cover2008}. From the results therein, the bound of  Proposition~\ref{prop8} is optimal when $P_2 \to 0$.
\end{remark} 
We now present our lower bound.
\begin{theorem} \label{thm:9}
A lower bound for  the problem of estimation with a helper that knows both the interference and the signal noncausally is given by
\begin{align*}
D(P)_{\rm min} \ge \frac{ (\frac{\alpha^2 P_1 P_2}{P_1 + \alpha^2 P_2})N}{MSE(\alpha)}
\end{align*}
for \textit{any} $\alpha \in \mathcal{R}$, $\alpha \neq 0$, where $MSE(\alpha)$ is given by the optimum value of the following convex (quadratic) optimization problem:
\begin{align*}
\max_{|\rho_{XS_1}| \le \sqrt{PP_1}, |\rho_{XS_2}| \le \sqrt{PP_2}} P + (1 -\alpha)^2 P_2 + 2(1-\alpha)\rho_{XS_2}+N - \frac{((1-\alpha)\alpha P_2 + \alpha \rho_{XS_2} + \rho_{X S_1})^2}{P_1 + \alpha^2 P_2}. 
\end{align*}
\end{theorem}
It can be shown that setting $\alpha =1$ and $\alpha \rightarrow \infty$ recovers the bounds in Propositions~\ref{prop7} and \ref{prop8}, respectively. The cases of $\alpha=1$ and $\alpha=\infty$ correspond to supplying $S_1+S_2$ and $S_2$, respectively to the decoder and then lower bounding the distortion.  

Note that while finding the optimum value of $\alpha$ may be difficult, Theorem~\ref{thm:9} gives a lower bound for every $\alpha$. We note also that while computation of the lower bound requires solving an optimization problem for each $\alpha$, unlike the lower bounds in Propositions~\ref{prop7} and~\ref{prop8}, the optimization problem is quadratic and can be efficiently solved~\cite{Boyd--Vandenberghe2004},~\cite{cvx}.

\begin{proof}
The idea in the proof of Theorem~\ref{thm:9} lies in giving side information $S_1 + \alpha S_2$ to the decoder instead of just $S_1+S_2$ or $S_2$ as in Propositions~\ref{prop7} and~\ref{prop8} respectively, and then a more careful bounding of the terms appearing in the distortion calculation using Linear minimum mean square error estimation and convex optimization. 

From the data processing inequality,
\begin{align*}
I(S_1^n;\Sh_1^n | S_1^n+ \alpha S_2^n) &\le I(S_1^n; Y^n|S_1^n + \alpha S_2^n) \\
& = h(Y^n|S_1^n+ \alpha S_2^n) - h(Z^n) \\
& \le \sum_{i=1}^n h(Y_i | S_{1i} + \alpha S_{2i}) - \frac{n}{2} \log 2\pi e N \\
& \stackrel{(a)}{\le} n h(Y | S_{1} + \alpha S_{2}, Q) - \frac{n}{2} \log 2\pi e N \\
& \le n h(Y | S_{1} + \alpha S_{2}) - \frac{n}{2} \log 2\pi e N.
\end{align*}
In $(a)$, we defined $Q \sim Unif[1:n]$ independent of all other random variables and $Y_Q = Y$, $S_{1Q} = S_1$, $S_{2Q} = S_2$ and $\Sh_{1Q} = \Sh_1$. On the other hand, we have{\allowdisplaybreaks
\begin{align*}
I(S_1^n; \Sh_1^n |S_1^n + \alpha S_2^n) & = \sum_{i=1}^n h(S_{1i}| S_{1i} + \alpha S_{2i}) - h(S_1^n |\Sh_1^n, S_1^n + \alpha S_2^n) \\
& \ge \sum_{i=1}^n h(S_{1i}| S_{1i} + \alpha S_{2i}) - \sum_{i=1}^n h(S_{1i} |\Sh_{1i}) \\
& \ge \sum_{i=1}^n h(S_{1i}| S_{1i} + \alpha S_{2i}) - \sum_{i=1}^n h(S_{1i} - \Sh_{1i}) \\
& \stackrel{(a)}{\ge} n h(S_1 | S_1 + \alpha S_2) - \frac{n}{2} \log 2\pi e D(P)_{\rm min}\\
& = \frac{n}{2}\log\left(2\pi e \frac{\alpha^2 P_1 P_2}{P_1 + \alpha^2 P_2}\right)- \frac{n}{2} \log 2\pi e D(P)_{\rm min},
\end{align*}}
where $(a)$ follows from concavity of differential entropy and the property that a Gaussian distribution maximizes the differential entropy for a given second moment. Therefore, 
\begin{align*}
\frac{1}{2}\log\left(2\pi e \frac{\alpha^2 P_1 P_2}{P_1 + \alpha^2 P_2}\right)- \frac{1}{2} \log 2\pi e D(P)_{\rm min} &\le h(Y | S_{1} + \alpha S_{2}) - \frac{1}{2} \log 2\pi e N \\
& = h(X+ (1-\alpha)S_2 +Z | S_1 + \alpha S_2) - \frac{1}{2}\log 2\pi e N \\
& \le h(X+ (1-\alpha)S_2 +Z - k (S_1 + \alpha S_2)) - \frac{1}{2}\log 2\pi e N,
\end{align*}
where $k$ is defined as
\begin{align*}
k := \frac{(1-\alpha)\alpha P_2 + \alpha \rho_{XS_2} + \rho_{XS_1}}{P_1 + \alpha^2 P_2},
\end{align*}
with $\E XS_1: = \rho_{XS_1}$ and $\E XS_2: = \rho_{XS_2}$. From Cauchy-Schwartz inequality and the power constraint on $X$, $|\rho_{XS_1}| \le \sqrt{PP_{S_1}}$ and $|\rho_{XS_2}| \le \sqrt{PP_{S_2}}$.

Continuing with our bound, we have 
\begin{align*}
h(X+ (1-\alpha)S_2 +Z - k (S_1 + \alpha S_2))  \le \frac{1}{2}\log (2 \pi e (\E(X+ (1-\alpha)S_2 +Z - k (S_1 + \alpha S_2))^2)).
\end{align*}
In turn, we have
\begin{align*}
\E(X+ (1-\alpha)S_2 +Z - k (S_1 + \alpha S_2))^2 &= P + (1 -\alpha)^2 P_2 + 2(1-\alpha)\rho_{XS_2}\\
& \quad +N - \frac{((1-\alpha)\alpha P_2 + \alpha \rho_{XS_2} + \rho_{X S_1})^2}{P_1 + \alpha^2 P_2} \\
& := MSE(\alpha, \rho_{XS_1}, \rho_{XS_2}).
\end{align*}
Note now that for $\alpha$ fixed, $MSE(\alpha, \rho_{XS_1}, \rho_{XS_2})$ is a concave (quadratic) function of $\rho_{XS_1}$ and $\rho_{XS_2}$, and the constraints $|\rho_{XS_1}| \le \sqrt{PP_{S_1}}$ and $|\rho_{XS_2}| \le \sqrt{PP_{S_2}}$ are linear constraints. Hence, we can find the maximum value using convex optimization. Letting $\rho_{XS_1}^*$ and $\rho_{XS_2}^*$ denote the optimal solutions to the optimization problem, we arrive at the lower bound for the achievable distortion:
\begin{align*}
D(P)_{\rm min} \ge \frac{ (\frac{\alpha^2 P_1 P_2}{P_1 + \alpha^2 P_2})N}{MSE(\alpha, \rho_{XS_1}^*, \rho_{XS_2}^*)}.
\end{align*}
\end{proof}

\subsection*{Comparison of bounds}
As we mentioned earlier, Theorem~\ref{thm:9} includes the bounds in Proposition~\ref{prop7} and~\ref{prop8}. It can also be larger, as we now show in an example. 

Let $P_1 = 1$, $N = 1$ and $P =1$. We vary $P_2$ and compare the bounds obtained with different values of $P_2$. The plots comparing the various upper and lower bounds are given in Figure~\ref{fig:svfig1}.
\begin{figure}[!h]
\psfrag{Proposition 7}[l]{Proposition \ref{prop7}}
\psfrag{Proposition 8}[l]{Propostion \ref{prop8}}
\psfrag{Theorem 9}[l]{Theorem \ref{thm:9}}
\psfrag{Upper bound}[l]{Theorem \ref{thm:8}}
\begin{center}
\scalebox{0.7}{\includegraphics{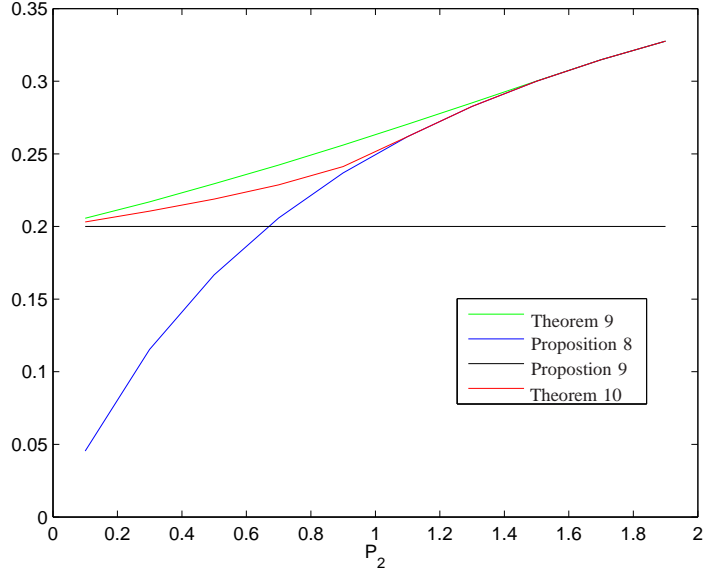}}
\caption{Comparison of bounds for estimation with a helper that knows both the interference and the source. This figure gives a plot of the various bounds on distortion for different values of $P_2$} \label{fig:svfig1}
\end{center}
\end{figure}
As can be seen from Figure~\ref{fig:svfig1}, the lower bound given by Theorem~\ref{thm:9} can be strictly better than that given by previous lower bounds. As we noted in the proof of Theorem~\ref{thm:9}, the improvement comes from two aspects: giving $S_1+ \alpha S_2$ to the decoder and a more careful bounding via Linear Minimum Mean Square Error Estimation and Convex Optimization. The reader may ask whether it is necessary to use $S_1 +\alpha S_2$ instead of just setting $\alpha =1$ or $\alpha \to \infty$ and calculate the bounds more carefully using Linear Estimation and Convex Optimization. In our simulation, we noted that for some values of $P_2$, moderate values of $\alpha$, such as $\alpha =2,3$ give better bounds than $\alpha = 1$ or $\alpha = 20$. This shows that using $S_1 + \alpha S_2$ does lead to better bounds than using $S_1 + S_2$ or $S_2$ alone. 

\subsection{Constant gap results}
In our simulations, we note that the upper bound and lower bounds appear to be quite close. This suggests that constant multiplicative gap results on the distortion may be possible, under some conditions on the input, source and interference powers. This is indeed the case as stated in our next result that when the interference power is larger than a threshold (that depends on the system parameters), the lower and upper bounds are within a constant multiplicative gap.
\begin{theorem} \label{thm:10}
If 
\begin{align}
\sqrt{P_2} \geq \frac{\sqrt{\gamma^2P+\gamma\sqrt{ P}(2+\gamma\sqrt{ P})(P(1-\gamma^2)+N)}-\gamma\sqrt{P}}{\gamma\sqrt{ P}(2+\gamma\sqrt{ P})}, \label{eqn:term2}
\end{align}
with $\gamma = \sqrt{\frac{\epsilon (P+N)}{2P}}$, $0\leq\epsilon\leq\frac{P}{P+N}$, then the multiplicative gap between the upper bound in Theorem~\ref{thm:8}, $D_{\rm achievable}$, and the lower bound in Proposition~\ref{prop7}, $D_{\rm lb}$, is at most $1/(1-\e)$. That is,  
\begin{align*}
\frac{D_{\rm achievable}}{D_{\rm lb}} \leq \frac{1}{1-\epsilon}.
\end{align*}
\end{theorem}
\begin{proof}
We begin the proof by evaluating the distortion achieved by Theorem~\ref{thm:8} for $\alpha=-\beta= \sqrt{\frac{\epsilon (P+N)}{2P}}$. We have 
\begin{equation}\label{eqn:term1}
1+\frac{P(1-\alpha^2-\beta^2)}{N} = 1+\frac{P}{N}\left(1-\frac{\epsilon(P+N)}{P}\right)=(1-\epsilon) \left(1+\frac{P}{N}\right).
\end{equation}
Now from the condition on $P_2$ stated in the Theorem (see \eqref{eqn:term2}), it follows that, 
\begin{align*}
P_2\alpha\sqrt{P}(2+\alpha\sqrt{P})+\sqrt{P_2}2\alpha\sqrt{P}-P(1-\alpha^2)-N&\geq 0\\
\Rightarrow P_2\alpha\sqrt{P}(2+\alpha\sqrt{P})+\sqrt{P_2}(2+\alpha\sqrt{P})\alpha\sqrt{P}-\alpha^2P-P(1-2\alpha^2)-N&\geq 0\\
\Rightarrow P_2\left[\alpha\sqrt{P}(2+\alpha\sqrt{P})-\frac{\alpha^2P}{\sqrt{P_2}}+\frac{(2+\alpha\sqrt{P})\alpha\sqrt{P}}{\sqrt{P_2}}-\frac{\alpha^2P}{P_2}\right]&\geq P(1-2\alpha^2)+N\\
\Rightarrow P_2\left[\left(\alpha\sqrt{P}+\frac{\alpha\sqrt{P}}{\sqrt{P_2}}\right)\left(2+\alpha\sqrt{P}-\frac{\alpha\sqrt{P}}{\sqrt{P_2}}\right)\right] &\geq P(1-2\alpha^2)+N\\
\Rightarrow P_2\left[\left(1+\alpha\sqrt{P}\right)^2-\left(1-\alpha\sqrt{\frac{P}{P_2}}\right)^2\right] &\geq P(1-2\alpha^2)+N\\
\Rightarrow \frac{(\alpha\sqrt{P}+1)^2}{(-\alpha\sqrt{\frac{P}{P_2}}+1)^2P_2+P(1-2\alpha^2)+N} &\geq \frac{1}{P_2}.
\end{align*}

Therefore we have, 
\begin{equation*}
\left(1+\frac{(\alpha\sqrt{P}+1)^2P_{1}}{(\beta\sqrt{\frac{P}{P_2}}+1)^2P_2+N+P(1-\alpha^2-\beta^2)}\right)\left(1+\frac{P(1-\alpha^2-\beta^2)}{N}\right) \geq \left(1+\frac{P_1}{P_2}\right)\left(1+\frac{P}{N}\right)(1-\epsilon),
\end{equation*}
which implies 
\begin{equation*}
\frac{D_{\rm achievable}}{D_{\rm lb}} \leq \frac{1}{1-\epsilon}.
\end{equation*}
\end{proof}
\section{Conclusion}
In this paper, we defined and analyze the problem of estimation with a helper that knows the interference. In the discrete memoryless case when the interfering signal, $S_2$, is known causally at the encoder, we characterized the distortion-cost region. When $S_2$ is known noncausally, we proposed an achievability scheme based on hybrid coding. In the binary estimation with a helper problem, we also proposed two lower bounds. Using the upper and lower bounds, we characterized the distortion-cost region when the problem parameters $C$, $p_1$ and $p_2$ satisfy one of several nontrivial conditions. 

In the Gaussian case, we derived a lower bound based on a recent result by Verd\'{u} between divergence and mismatched estimation. We showed through numerical simulations that this lower bound can be strictly better than previous lower bound derived in \cite{Grover--Sahai2011}. Similar to the binary case, we also characterized the distortion-cost region when the problem parameters $P$, $P_1$ and $P_2$ satisfy one of several conditions.

We also extended our analysis in the Gaussian case to consider the case when the helper knows both $S_1$ and $S_2$ noncausally. In this case, we derived a lower bound that contains previous lower bounds proposed in~\cite{Huang--Narayanan2011} and can be strictly better. We also obtained constant multiplicative gap results for this setting.

In deriving our lower bound for the Gaussian case when only the interfering signal is known at the helper, we used a relationship between mismatched estimation and divergence. In the discrete case, a relationship between divergence and Hamming distortion exists too. One such relationship is Marton's inequality~\cite[Lemma 6.3]{Gray2010}. An interesting open question is whether one can use such relationships to derive a lower bound for the binary case that is strictly better than the bounds we proposed. 


\section*{Acknowledgment}
We thank Mr Gowtham Kumar of Stanford University for discussions that motivated this work, and Professor Sriram Vishwanath of The University of Texas at Austin for helpful discussions during the course of this work. 

\bibliographystyle{IEEEtran}
\bibliography{estimate}

\appendices
\section{Sketch of Achievability for Theorem~\ref{thm:1}} \label{appen_thm1}
We use block Markov coding over $B$ blocks. The scheme in each block is basically a separation scheme, where we use the random variable $U$ for transmission of a message from the previous block. The message itself is a Wyner-Ziv description~\cite{Wynerrd} of $S_2^n$ from the previous block. More concretely, in each block $j \in [1:B]$, the \textit{transmission} codebook is generated as follows: Generate $2^{n(I(U;Y) - \e)}$ $U^n(l)$ sequences according to $\prod_{i=1}^n p(u_i)$. The \textit{compression} codebook is generated by the following two step procedure: Generate $2^{n(I(V;S_2, U) + \e)}$ $V^n$ sequences according to $\prod_{i=1}^n p(v_i)$. Partition the set of $V^n$ sequences into $2^{n (I(V;S_2|U,Y) + 2\e)}$ bins, $\Bc(M_j)$. 

For encoding, at the end of block $j$, assume that the codeword $U^n(m_j)$ was sent. The encoder then finds a $V^n(j)$ sequence that is jointly typical with $(U^n(m_j), S_2^n(j))$. If there is more than one such sequence, it picks from one uniformly at random from the set of jointly typical sequences. This operation succeeds with high probability as $n \to \infty$ since there are $2^{n(I(V;S_2, U) + \e)}$ $V^n(j)$ sequences. The encoder then finds the bin index $M_{j+1}$ such that $V^n \in \Bc(M_{j+1})$. It then sends out the index $M_{j+1}$ in block $j+1$ by selecting $U^n(j+1)$ and sending out the $X^n$ sequence encoded as $x_i = f(u_i(M_{j+1}), s_{2i}(j+1))$. For the first block, the encoder sends an arbitrary message. This encoding operation requires the condition that
\begin{align*}
I(U;Y) - \e > I(V;S_2|U,Y) + 2\e.
\end{align*}

For decoding, at the end of block $j+1$, the decoder first decodes the bin index $M_{j+1}$. From standard arguments (see for e.g.~\cite[Chapter 7]{El-Gamal--Kim2010}), this decoding operation succeeds with high probability provided
\begin{align*}
I(U;Y) - \e > I(V;S_2|U,Y) + 2\e.
\end{align*} 

Once the decoder recovers the bin index $M_{j+1}$, it then recovers the true $V^n(j)$ codeword by looking for $v^n(j) \in \Bc(M_{j+1})$ such that $(u^n(m_j), y^n(j), v^n(j)) \in \aep$. It then reconstructs $S^n_1(j)$ as $\sh_1(u_i(m_j), y_i(j), v_i(j))$ for $i \in [1:n]$. From the rates given and standard arguments (see~\cite[Chapter 3 and Chapter 11]{El-Gamal--Kim2010}), the expected distortion for $S_1^n(j)$ in block $j$ is less than or equal to $\E d(S_1, \Sh_1(U,V,Y)) + \d(\e)$, where $\d(\e) \to 0$ as $\e \to 0$. This decoding and reconstruction procedure applies for the first $B-1$ blocks and for the $B$th block, we simply reconstruct $S_1^n(B)$ according to an arbitrary symbol $\sh_1 \in \mathcal{\Sh}_1$, incurring a distortion that is bounded by $D_{\rm max}$, where $D_{\rm \max} := \max_{\sh_1} \E d(S_1, \sh_1)$. The per symbol distortion over $B$-blocks is then upper bounded by $D + \d'(\e)$ where $\d'(\e) \to 0$ as $\e \to 0$.

We now note that the above achievability scheme takes care of the case when $I(U;Y) > I(V;S_2 |U,Y)$. The boundary case of $I(U;Y) = I(V;S_2 |U,Y)$ can be handled as follow. Assume first that $I(U;Y) >0$. Define $U' = (U,Q)$, $Q \in \{1,2\}$ independent of other random variables, $V' = V$ when $Q = 1$ and $V = \emptyset$ when $Q = 2$. $X = f(U,S_2)$ regardless of $Q$ and $\sh_1(U',V',Y') = \sh_1(U,V,Y)$ if $Q = 1$ and $\sh^*_1$ if $Q = 2$, where $\sh^*_1$ is an arbitrary symbol belonging to $\mathcal{\Sh}_1$. Let $\P(Q = 1) = p_1$. We have
\begin{align*}
I(U';Y') &\ge I(U;Y), \\
I(V';S_2 |U',Y') &= p_1I(V;S_2|U,Y), \\
\E d(S_1, \Sh_1(U',V',Y')) &\le p_1 \E d(S_1, \Sh_1(U,V,Y)) +(1-p_1)D_{\rm max}.  
\end{align*}
With this choice of random variables, $I(U';Y') > I(V'; S_2|U', Y')$ whenever $p_1 <1$ and we can then apply the achievability scheme we discussed, at the expense of larger distortion. By choosing $p_1(n) = 1- \e_n$, where $\e_n \to 0$ as $n \to \infty$, we can apply our achievability scheme for blocklength $n$ sufficiently large, with the resulting expected distortion converging to $\E d(S_1, \Sh_1(U,V,Y))$ as $n \to \infty$. 

For the case of $I(U;Y) = I(V;S_2|U,Y) = 0$, it can be shown that in this case, the decoder can perform the reconstruction based only on $\sh_1(Y_i, U_i)$ for $i \in [1:n]$. Achievability in this case requires no block Markov coding. We only need to generate one transmission codeword $U^n$ and transmit $X^n$ according to $x_i = f(u_i, s_{2i})$. The decoder reconstructs $S_1^n$ as $\sh_1(u_i, y_i)$ for $i \in [1:n]$.

\section{Proof of Claim~\ref{clm1}} \label{appen_claim}
The causal region in Theorem~\ref{thm:1} is given by 
\begin{align*}
&\min \quad \E d(S_1, \Sh_1(U,V,Y)) \\
&\mbox{subject to }\\
& \quad I(U;Y) \ge I(V;S_2|U,Y) \\
& \quad \E X \le C
\end{align*}
for some $p(u)p(v|u,s_2)$ and function $x(u,s_2)$. We prove that $D(0.11)_{\rm min-causal}>0$ by contradiction. Suppose that there exists $U,V$ satisfying the constraints such that $ \E d(S_1, \Sh_1(U,V,Y))  = 0$. This implies in particular that $H(S_1 |U,V,Y) = 0$. Hence,
\begin{align*}
 I(V;S_2|U,Y) & =  I(V, S_1;S_2|U,Y) \\
 & \ge  I(S_1;S_2|U,Y) \\
 &  = H(S_1 |U,Y) - H(S_1|S_2, U,Y) \\
 & =  H(S_1 |U,Y) \\
 & = H(S_1,Y|U) - H(Y|U) \\
 & = H(S_1) + H(Y|U,S_1) - H(Y|U).
\end{align*}
The last step follows from $U$ being independent of $S_1$. Since we require $I(U;Y) \ge I(V;S_2|U,Y)$, and we know that $ H(S_1) + H(Y|U,S_1) - H(Y|U)\le I(V;S_2|U,Y)$, a necessary condition for $ \E d(S_1, \Sh_1(U,V,Y))  = 0$ is
\begin{align*}
&H(S_1) + H(Y|U,S_1) - H(Y|U) \le I(U;Y), \\
&\Rightarrow H(S_1) + H(Y|U,S_1)  \le H(Y).
\end{align*}
Define the subsets of $\Uc$ as follows. $\Uc_0 := \{u: x(u,s_2) = 0\}$; $\Uc_1 := \{u: x(u,s_2) = 1\}$; $\Uc_s := \{u: x(u,s_2) = s_2\}$; and $\Uc_{\bar{s}} := \{u: x(u,s_2) = 1\oplus s_2\}$. Note the following.
\begin{itemize}
\item For $u \in \Uc_0$, $H(Y|U = u,S_1) = 1$ since $S_2$ is independent of $U,S_1$.
\item For $u \in \Uc_1$, $H(Y|U = u,S_1) = 1$ since $S_2$ is independent of $U,S_1$.
\item For $u \in \Uc_s$, $H(Y|U = u,S_1) = 0$ since $S_2\oplus X = 0$ and $Y = S_2\oplus X \oplus S_1$.
\item For $u \in \Uc_{\bar{s}}$, $H(Y|U = u,S_1) = 0$ since $S_2\oplus X = 1$ and $Y = S_2\oplus X \oplus S_1$.
\end{itemize}
Further, define $p_{u0} = \sum_{u \in \Uc_0} p(u)$; $p_{u1} = \sum_{u \in \Uc_1} p(u)$; $p_s = \sum_{u \in \Uc_s} p(u)$; and $p_{\bar{s}} = \sum_{u \in \Uc_{\bar{s}}} p(u)$. Then,
\begin{align*}
H(S_1) + H(Y|U,S_1) &= H_2(p_1) + p_{u0} + p_{u1} \\
& = H_2(p_1) + 1-C_s,
\end{align*}
where $C_s = p_s + p_{\bar{s}}$. 

The cost constraint can be expressed as
\begin{align*}
\E X &= p_1 + \frac{1}{2} (p_s + p_{\bar{s}}) \\
& = p_1 + \frac{1}{2} C_s \\
& \le C,
\end{align*}
where $C = 0.11$. In particular, the cost constraint implies that $C_s \le 2C$. Hence,
\begin{align*}
H(S_1) + H(Y|U,S_1) \ge 1+H_2(q) - 2C.
\end{align*} 
Now, since $p_1 = 0.1$ and $C = 0.11$, we see that
\begin{align*}
H(S_1) + H(Y|U,S_1) &> 1 \\
& \ge H(Y),
\end{align*}
which is a contradiction. 
\section{Derivation of Theorem~\ref{thm:5}} \label{appen_g1}
The derivation of Theorem~\ref{thm:5} follows from choosing the auxiliary random variables in Theorem~\ref{thm:2}. 
Starting from Theorem~\ref{thm:2}, let
\begin{align*}
U = X' + \gamma S_2, \\
X = \alpha \sqrt{\frac{P}{P_2}}S_2 + X',\\
X' \sim N(0, P'),\\
\E(S_2 X') = \beta \sqrt{P_2 P'},
\end{align*}
where $P'$ is a quantity to be calculated, and $\alpha$ and $\beta$ are restricted to be between -1 to 1 to satisfy the power constraints. Observe that $X$ is a function of $U, S_2$ as required. For convenience, we will use the notation $X|Y$ to denote Minimum Mean Square Error of $X$ given $Y$. The reconstruction function is given by 
\begin{align*}
\Sh_1 = \E(S_1 |U,Y).
\end{align*}

We now determine $P'$ from other variables using $\E X^2 = P$.
\begin{align*}
\E X^2 = \alpha^2 P + P' + 2\alpha \beta \sqrt{P P'} = P.
\end{align*}

Solving for $P'$ gives
\begin{align*}
\sqrt{P'} = \frac{-2\alpha \beta \sqrt{P} + \sqrt{4\alpha^2 \beta^2 P + 4(1-\alpha^2)P}}{2}.
\end{align*}

To satisfy the constraint in Theorem~\ref{thm:2}, we require
\begin{align*}
h(U|S_2) > h(U|Y).
\end{align*}
Since $U,S_2,Y$ are all Gaussian random variables, this condition reduces to
\begin{align*}
U|S_2 > U|Y.
\end{align*}
Now, 
\begin{align*}
U|S_2 &= X'|S_2\\
& = (1-\beta^2)P'.
\end{align*}
As for $U|Y$, we have
\begin{align*}
U|Y = \E U^2 - \frac{(\E (UY))^2}{\E Y^2},
\end{align*}
and
\begin{align*}
\E U^2 &= P' + 2\gamma \beta \sqrt{P' P_2} + \gamma^2 P_2, \\
\E(UY) &= \E((X' + \gamma S_2)(\alpha \sqrt{\frac{P}{P_2}} S_2 + X')) + \beta\sqrt{P' P_2} + \gamma P_2 \\
& = \alpha \beta \sqrt{P P'} + P' + \alpha \gamma \sqrt{P P_2} + \gamma \beta \sqrt{P' P_2}+ \beta\sqrt{P' P_2} + \gamma P_2,\\ 
\E Y^2 &= \E X^2 + \E S_1^2 + \E S_2^2 + 2 \E (S_2 (\alpha \sqrt{\frac{P}{P_2}} S_2 + X')) \\
& = P + 1 + P_2  + 2 \alpha \sqrt{P P_2} + 2\beta \sqrt{P' P_2}.
\end{align*}
The expected distortion is then given by is $S_1|(Y,U)$, which is
\begin{align*}
S_1|(U,Y) &= 1 - [\E(US_1) \E(YS_1) ] \left[\begin{array}{cc} \E U^2 & \E(UY) \\ \E(UY) & \E Y^2 \end{array}\right]^{-1}\left[ \begin{array}{c} \E (US_1) \\ \E(Y S_1)\end{array}\right].
\end{align*}

We note now that $\E (US_1) = 0$ and $\E (S_1Y) = 1$. The lower bound therefore works out to
\begin{align*}
S_1|(U,Y) = 1 - \frac{\E U^2}{\E Y^2 \E U^2 - (\E UY)^2}.
\end{align*}

\section{Proof of Theorem~\ref{thm:7}} \label{appen_glb}
As Theorem~\ref{thm:7} is a generalization of Proposition~\ref{prop6}, the proof of this Theorem also follows closely that of Proposition~\ref{prop6}. As such, we will only mention areas where there are differences from the proof in Proposition~\ref{prop6} and refer readers to Proposition~\ref{prop6} for the rest of the proof.

As we mentioned before, we generalize the bound by not assuming that $X^n \equiv 0$. Instead, let us assume that under the mismatched distribution $Q$, $X$ is distributed i.i.d according to $X = cS_2 + Z$, where $Z \sim N(0, rP)$ independent of $S_2$ and $r \ge 0$. Under this assumption, $MSE_Q(\gamma)$ and $D(P_{Y^n| S_2^n}^{(\gamma)}||Q_{Y^n| S_2^n}^{(\gamma)})$ used in the proof of Proposition~\ref{prop6} are now different. The bounds on $MSE_Q(\gamma)$ and the divergence between the true distribution and the mismatched distribution are therefore different. We calculate them as follow.

Define $\alpha$ as
\[
\alpha := \frac{1}{1+ \gamma P_I},
\]
where $P_I = (1+c)^2P_2 + rP$. Let $\E||X^n||_2^2 = nx^2$, where $x^2 \le P$. We now have, for $MSE_Q(\gamma)$, 
\begin{align*}
\frac{1}{n}MSE_Q(\gamma)  &= \frac{1}{n}\E||S_1^n - \alpha(X^n + S_1^n + S_2^n)||^2 \\
&= \frac{1}{n}\E||S_1^n - \alpha(S_1^n + S_2^n) ||^2 - \frac{2 \alpha}{n} \E <S_1^n - \alpha(S_1^n+S_2^n), X^n> + \frac{\alpha^2}{n} \E||X^n ||^2 \\
& = \frac{(1-\alpha)^2}{\gamma} + \alpha^2P_2  + \frac{2 \alpha^2}{n} \E <S_2^n, X^n> + \alpha^2 x^2 \\
&=\frac{\gamma P_I^2}{(1+ \gamma PI)^2}+ \frac{P_2}{(1+ \gamma P_I)^2}  + \frac{2 \alpha^2}{n} \E <S_2^n, X^n> + \alpha^2 x^2 \\
& = \frac{P_I}{1+ \gamma P_I} - \frac{P_I}{(1+ \gamma P_I)^2} +  \frac{P_2}{(1+ \gamma P_I)^2}+ \frac{2}{n(1+ \gamma P_I)^2} \E <S_2^n, X^n> + \frac{1}{(1+ \gamma P_I)^2} x^2.
\end{align*}
It remains to calculate an upper bound on the divergence. As before, $P_{Y^n| S_2^n}^{(\gamma)} \sim N(S_2^n + X^n, \frac{1}{\gamma}I)$, but now, $Q_{Y^n| S_2^n}^{(\gamma)} \sim N((1+c)S_2^n, (\frac{1}{\gamma} + rP)I)$. The (conditional) divergence is now given by
\begin{align*}
\frac{2}{n}D(P_{Y^n| S_2^n}^{(\gamma)}||Q_{Y^n| S_2^n}^{(\gamma)}) &= \frac{1}{1+ \gamma rP} -1 - \log(\frac{1}{1+ \gamma rP}) + \frac{ \gamma}{n(1+ \gamma rP)} ||X^n - cS_2^n||_2^2.
\end{align*}
Combining the divergence bound after taking expectation over $S_2^n$ with the $MSE_Q$ bound after integration gives (see \eqref{eqn:1} in the proof of Proposition~\ref{prop6}) {\allowdisplaybreaks
\begin{align*}
\frac{\gamma_1 - \gamma_0}{n}MMSE (\gamma_0) &\ge \log(\frac{1+ \gamma_1 P_I}{1+ \gamma_0 P_I}) + \frac{1}{(1+ \gamma_1 P_I)} - \frac{1}{(1+ \gamma_0 P_I)} -\frac{P_2}{P_I(1+ \gamma_1 P_I)}\\
& \quad +\frac{P_2}{P_I(1+ \gamma_0 P_I)} - \frac{c^2\gamma_1}{1+\gamma_1rP}P_2  -\frac{1}{1+ \gamma_1 rP} +1 + \log(\frac{1}{1+ \gamma_1 rP}) \\
& \quad + (\frac{1}{P_I (1+ \gamma_0 P_I)} - \frac{1}{P_I (1+ \gamma_1 P_I)} - \frac{\gamma_1}{1+\gamma_1rP})x^2 \\
&\quad + 2(\frac{1}{P_I (1+ \gamma_0 P_I)} - \frac{1}{P_I (1+ \gamma_1 P_I)} +  \frac{c\gamma_1}{1+\gamma_1rP})\frac{\E<S_2^n, X^n>}{n} \\
&\ge \log(\frac{1+ \gamma_1 P_I}{1+ \gamma_0 P_I}) + \frac{1}{(1+ \gamma_1 P_I)} - \frac{1}{(1+ \gamma_0 P_I)} -\frac{P_2}{P_I(1+ \gamma_1 P_I)}\\
& \quad +\frac{P_2}{P_I(1+ \gamma_0 P_I)} - \frac{c^2\gamma_1}{1+\gamma_1rP}P_2 -\frac{1}{1+ \gamma rP} +1 + \log(\frac{1}{1+ \gamma rP}) \\
& \quad + (\frac{1}{P_I (1+ \gamma_0 P_I)} - \frac{1}{P_I (1+ \gamma_1 P_I)} - \frac{\gamma_1}{1+\gamma_1rP})x^2 \\
&\quad - |2(\frac{1}{P_I (1+ \gamma_0 P_I)} - \frac{1}{P_I (1+ \gamma_1 P_I)} +  \frac{c\gamma_1}{1+\gamma_1rP})|\sqrt{P_2} |x|.
\end{align*}\allowdisplaybreaks}
The final line follows from successive application of Cauchy-Schwartz on $\E<S_2^n, X^n>$. Minimizing the bound over $|x| \le \sqrt{P}$ then gives the generalized lower bound. Let $a = (\frac{1}{P_I (1+ \gamma_0 P_I)} - \frac{1}{P_I (1+ \gamma_1 P_I)} - \frac{\gamma_1}{1+\gamma_1rP})$ and $b =  |2(\frac{1}{P_I (1+ \gamma_0 P_I)} - \frac{1}{P_I (1+ \gamma_1 P_I)} +  \frac{c\gamma_1}{1+\gamma_1rP})|\sqrt{P_2}$. We note that $b \le 0$ and let $f(x) = ax^2 - b|x|$.

We note that if $a\le 0$, $f(x)$ is symmetric and decreasing in $x$. Therefore, we set $x^* = \sqrt{P}$. If $a > 0$, then $x^* = b/(2a)$ if $b/(2a) < \sqrt{P}$ and $x^* = \sqrt{P}$ otherwise. The generalized lower bound is now given by
\begin{align*}
\frac{\gamma_1 - \gamma_0}{n}MMSE (\gamma_0) & \ge \log(\frac{1+ \gamma_1 P_I}{1+ \gamma_0 P_I}) + \frac{1}{(1+ \gamma_1 P_I)} - \frac{1}{(1+ \gamma_0 P_I)} -\frac{P_2}{P_I(1+ \gamma_1 P_I)}\\
& \quad +\frac{P_2}{P_I(1+ \gamma_0 P_I)} - \frac{c^2\gamma_1}{1+\gamma_1rP}P_2  -\frac{1}{1+ \gamma rP} +1 + \log(\frac{1}{1+ \gamma rP}) \\
& \quad + (\frac{1}{P_I (1+ \gamma_0 P_I)} - \frac{1}{P_I (1+ \gamma_1 P_I)} - \frac{\gamma_1}{1+\gamma_1rP})x^{*2} \\
&\quad - |2(\frac{1}{P_I (1+ \gamma_0 P_I)} - \frac{1}{P_I (1+ \gamma_1 P_I)} +  \frac{c\gamma_1}{1+\gamma_1rP})|\sqrt{P_2} x^*,
\end{align*}
where we optimize over $\gamma_1 \ge \gamma_0$, $r \ge 0$ and $c \in \mathcal{R}$. Noting that $MMSE(1)/n = D(P)_{\rm min}$ then completes the proof.

\section{Sketch of Theorem~\ref{thm:8}} \label{appen_gsv}
The achievability scheme in Theorem~\ref{thm:8} closely resembles \cite{Huang--Narayanan2011} and involves allocating a fraction of the power for transmitting a message (corresponding to a compressed version of the desired source $S_1$) using dirty paper coding and using the remaining power for uncoded transmission of a linear combination of $S_1$ and $S_2$. The compressed index is generated based on Wyner-Ziv coding and then transmitted reliably over channel using dirty paper coding as in \cite{Huang--Narayanan2011}. The bin indices in Wyner-Ziv coding are transmitted at a rate equal to the capacity of the dirty paper channel. Note that the interference in this channel also includes the signal due to uncoded transmission created at the encoder. The compressed index is decoded at the receiver using the receiver side information $Y$ and both the decoded codeword and $Y$ are used to estimate the source $S_1$. Uncoded transmission helps in improving the signal to noise ratio of the desired signal $S_1$ in $Y$. 

Let 
\begin{align*}
U &= X' + \left(\alpha\sqrt{\frac{P}{P_1}}+1\right)S_1 + \left(\beta \sqrt{\frac{P}{P_2}}+1\right)S_2\\
X &= X' + \alpha\sqrt{\frac{P}{P_1}}S_1 + \beta\sqrt{\frac{P}{P_2}}S_2\\
X' &\sim \mathcal{N}(0,P(1-\alpha^2-\beta^2)),\\
Y &= X+S_1+S_2+Z,
\end{align*}
where $X'$ is independent of $S_1$ and $S_2$ and corresponds to the coded part of the signal. Auxiliary $U$ is used to cancel the total interference to $X'$ as in dirty paper coding. The total interference is equal to $\left(\alpha\sqrt{\frac{P}{P_1}}+1\right)S_1 + \left(\beta \sqrt{\frac{P}{P_2}}+1\right)S_2$. As a result, a clean channel (without interference) is created between $X'$ and $Y$, which can be used to transmit the description of $S_1$ at a Wyner-Ziv rate equal to  $\frac{1}{2}\log\left(1+\frac{P(1-\alpha^2-\beta^2)}{N}\right)$. The received signal $Y$ can also be seen as a noisy version of the desired signal $S_1$, and is used along with the message transmitted to reconstruct $S_1$. 

Therefore, the resulting distortion in $S_1$ is given by 

\begin{equation*}
D = \frac{P_1}{\left(1+\frac{\left(\alpha\sqrt{\frac{P}{P_1}}+1\right)^2P_1}{ \left(\beta\sqrt{\frac{P}{P_2}}+1\right)^2P_2+P(1-\alpha^2-\beta^2)+N}\right)\left(1+\frac{P(1-\alpha^2-\beta^2)}{N}\right)}.
\end{equation*}

\end{document}